\documentclass[a4paper,USenglish,cleveref,numberwithinsect,authorcolumns]{lipics-v2021}

\usepackage{csquotes} % nicer quotes (TODO: is it needed?)
\usepackage{booktabs} % nicer tables

\usepackage{amsmath,amscd,amssymb,amsthm}
\usepackage{mathtools}
\usepackage{mathrsfs}
\usepackage{dsfont}
\usepackage{complexity}
\usepackage{stmaryrd}

\theoremstyle{definition}
\newtheorem{upright-definition}[theorem]{Definition}
\renewenvironment{definition}{\begin{upright-definition}}{\lipicsEnd\end{upright-definition}}

\nolinenumbers
\hideLIPIcs

\DeclarePairedDelimiter\floor{\lfloor}{\rfloor}
\DeclarePairedDelimiter\cardinality{\lvert}{\rvert}
\DeclarePairedDelimiter\abs{|}{|}
\providecommand\with{}
\newcommand\SetSymbol[1][]{\nonscript\,#1\vert \allowbreak \nonscript\,\mathopen{}}
\DeclarePairedDelimiterX\set[1]{\lbrace}{\rbrace}{ \renewcommand\with{\SetSymbol[\delimsize]} #1 }
\DeclarePairedDelimiterX\multiset[1]{\llbracket}{\rrbracket}{ \renewcommand\with{\SetSymbol[\delimsize]} #1 }
\DeclarePairedDelimiterX{\norm}[1]{\lVert}{\rVert}{#1}

\newcommand*{\Iff}{\Longleftrightarrow}

\newcommand*{\N}{\mathbb N}

\newcommand*{\Sup}{\mathrm{Sup}}
\newcommand*{\indic}{\mathds 1}

\newcommand*{\symdiff}{\triangle}
\newcommand*{\compl}[1]{\overline{#1}}
\renewcommand{\O}{\mathcal{O}}

\newcommand{\mincup}{\mathbin{\ooalign{$\cup$\cr\hss\raisebox{0.5ex}{\scriptsize $\downarrow$}\hss}}}
\newcommand{\bigmincup}{\mathop{\ooalign{$\displaystyle\bigcup$\cr\hidewidth{\Large$\downarrow$}\hidewidth\cr}}}

\makeatletter
\def\moverlay{\mathpalette\mov@rlay}
\def\mov@rlay#1#2{\leavevmode\vtop{%
   \baselineskip\z@skip \lineskiplimit-\maxdimen
   \ialign{\hfil$\m@th#1##$\hfil\cr#2\crcr}}}
\newcommand{\charfusion}[3][\mathord]{
    #1{\ifx#1\mathop\vphantom{#2}\fi
        \mathpalette\mov@rlay{#2\cr#3}
      }
    \ifx#1\mathop\expandafter\displaylimits\fi}
\makeatother

\newlang{\GVRS}{GVRS}
\newlang{\VRP}{VRP}
\newlang{\EVRP}{EVRP}
\newlang{\CVRP}{CVRP}
\newlang{\LoadCVRP}{LoadCVRP}
\newlang{\GasCVRP}{GasCVRP}
\newlang{\LoadGasCVRP}{LoadGasCVRP}
\newlang{\WRP}{WRP}
\newlang{\TSP}{TSP}
\newlang{\MetricTSP}{MetricTSP}
\newlang{\BinPacking}{BinPacking}
\newlang{\UnaryBinPacking}{UnaryBinPacking}
\newlang{\HetMdFpBinPacking}{HetMdFpBinPacking}
\newlang{\TrianglePacking}{TrianglePacking}
\newlang{\NTDM}{N3DMatching}
\newclass{\paraNP}{paraNP}
\renewcommand*{\P}{\mathcal P}
\newcommand*{\tw}{\mathrm{tw}}
\newcommand*{\opwidth}{\mathrm{width}}
\newcommand*{\opreduce}{\mathtt{reduce}}
\newcommand*{\oprmc}{\mathtt{rmc}}
\newcommand*{\opdetach}{\mathtt{detach}}
\newcommand*{\opglue}{\mathtt{glue}}
\newcommand*{\opproj}{\mathtt{proj}}
\newcommand*{\opjoin}{\mathtt{join}}
\newcommand*{\opins}{\mathtt{ins}}
\newcommand*{\opshift}{\mathtt{shift}}
\renewcommand*{\A}{\mathcal A}

\newcommand*{\T}{\mathcal T}
\newcommand*{\parP}{\mathcal P}
\newcommand*{\parQ}{\mathcal Q}
\newcommand*{\parR}{\mathcal R}

\title{Parameterized Complexity of Vehicle Routing}

\newcommand*{\hpiaffil}{Hasso Plattner Institute, University of Potsdam, Germany}
\author{Michelle Döring}{\hpiaffil \and \url{https://temporalgraph.notion.site/michelle-doering}}{michelle.doering@hpi.uni-potsdam.de}{https://orcid.org/0000-0001-7737-3903}{German Federal Ministry for Education and Research (BMBF) through the project ``KI Servicezentrum Berlin Brandenburg'' (01IS22092)}
\author{Jan Fehse}{\hpiaffil}{jan.fehse@student.hpi.uni-potsdam.de}{}{}
\author{Tobias Friedrich}{\hpiaffil}{tobias.friedrich@hpi.uni-potsdam.de}{}{}
\author{Paula Marten}{\hpiaffil}{paula.marten@student.hpi.uni-potsdam.de}{https://orcid.org/0009-0002-9121-2449}{}
\author{Niklas Mohrin}{\hpiaffil}{niklas.mohrin@student.hpi.uni-potsdam.de}{https://orcid.org/0009-0004-8296-0167}{}
\author{Kirill Simonov}{Department of Informatics, University of Bergen, Norway}{kirill.simonov@uib.no}{https://orcid.org/0000-0001-9436-7310}{}
\author{Farehe Soheil}{\hpiaffil}{farehe.soheil@hpi.uni-potsdam.de}{https://orcid.org/0000-0002-0504-8834}{}
\author{Jakob Timm}{\hpiaffil}{jakob.timm@student.hpi.uni-potsdam.de}{https://orcid.org/0009-0009-6007-0389}{}
\author{Shaily Verma}{\hpiaffil}{Shaily.Verma@hpi.de}{https://orcid.org/0009-0000-6789-1643}{}

\authorrunning{Döring, Fehse, Friedrich, Marten, Mohrin, Simonov, Soheil, Timm, and Verma}

\Copyright{Michelle Döring, Jan Fehse, Tobias Friedrich, Paula Marten, Niklas Mohrin, Kirill Simonov, Farehe Soheil, Jakob Timm, and Shaily Verma}

 %TODO mandatory: Please choose ACM 2012 classifications from https://dl.acm.org/ccs/ccs_flat.cfm 


\begin{document}
\maketitle

\begin{abstract}
    The Vehicle~Routing~Problem ($\VRP$) is a popular generalization of the Traveling Salesperson Problem.
Instead of one salesperson traversing the entire weighted, undirected graph $G$, there are $k$ vehicles available to jointly cover the set of clients $C \subseteq V(G)$.
Every vehicle must start at one of the depot vertices $D \subseteq V(G)$ and return to its start.
Capacitated~Vehicle~Routing ($\CVRP$) additionally restricts the route of each vehicle by limiting the number of clients it can cover, the distance it can travel, or both.

In this work, we study the complexity of $\VRP$ and the three variants of $\CVRP$ for several parameterizations, in particular focusing on the treewidth of $G$.
We present an $\FPT$ algorithm for $\VRP$ parameterized by treewidth.
For $\CVRP$, we prove $\paraNP$- and $\W[\cdot]$-hardness for various parameterizations, including treewidth, thereby rendering the existence of $\FPT$ algorithms unlikely.
In turn, we provide an $\XP$ algorithm for $\CVRP$ when parameterized by both treewidth and the vehicle capacity.

\end{abstract}

%\tableofcontents

\section{Introduction}

Optimizing logistics has been an important driver of the theory of computing since its very dawn, and, with the ever-growing digitalization and decentralization of services, it remains a productive direction to this day.
One of the canonical challenges in the area, Vehicle Routing Problem (VRP), introduced by George Dantzig and John Ramser in 1959~\cite{dantzig59},
can be informally summarized as follows:
``What is the optimal set of routes for a fleet of vehicles to traverse in order to deliver to a given set of customers?''

To model this question, it is natural to use a weighted graph to represent the relevant locations and the travel costs between them, with particular vertices marked as the depots, where the vehicles and the goods are originally located, and some other vertices are marked as clients that await the delivery. Formally, we consider $\VRP$ to be the following computational problem: Given an edge-weighted undirected graph $(G, w)$, subsets $C, D \subseteq V(G)$, and integers $k$, $r$, determine whether there exists a collection of at most $k$ closed walks in $G$, such that each walk starts and ends at a vertex of $D$, the walks cover all vertices in $C$, and their total weight is at most $r$.

Defined in this way, $\VRP$ also becomes interesting from the theoretical point of view, as it generalizes several important \NP-hard problems. Notably, the famous Traveling Salesperson Problem is a special case of $\VRP$ with one vehicle and every vertex being a client, i.e., $k = 1$, $C = D = V(G)$. On the other hand, since $\VRP$ allows for multiple vehicles, it also captures packing and covering problems on graphs. For example, the \textsc{Maximum Cycle Cover} problem, where the task is to find a collection of $k$ vertex-disjoint cycles in the given graph $G$ that covers all vertices~\cite{cygan2015parameterized}, is modeled by $\VRP$ with $C = D = V(G)$ and $r = |V(G)|$.

Based on the immediate hardness of the problem in the classical sense, practical solvers for vehicle routing are mostly based on heuristics, which do not provide guarantees either on the quality of the solution or on the running time. We refer to surveys by Braekers et al.~\cite{braekers2016vehicleSurvey} and by Mor and Speranza~\cite{mor2022vehicleSurvey} for an in-depth introduction to the vast area of practical research on solving the vehicle routing problem, as well as  for listing various variants of the vehicle routing problem that received practical interest.
Stemming from numerous research works and the efficiency of the developed solvers, the vehicle routing problem itself sees applications to other domains, for example, routing traffic in computer networks~\cite{akhoondian2020walking}.

On the other hand, theoretical results regarding $\VRP$ have been much scarcer, especially with respect to solving the problem exactly. The immediate \NP-hardness of the problem together with the diversity of explicit (such as number of vehicles and tour lengths) and implicit parameters (such as treewidth that may be small in certain applications~\cite{akhoondian2020walking}) makes $\VRP$ a natural target for parameterized complexity. However, to the best of our knowledge, this avenue was largely unexplored until now, although $\TSP$ itself~\cite{bodlaender2015deterministic,EuclideanTSPETH} and other related problems such as cycle packing~\cite{BentertFGKLP00S25} are well-studied in the area.
%Some variants of the problem have also been studied from the viewpoint of approximation algorithms.

\subparagraph*{Our contribution.} In this work, we aim to close this gap and initiate the parameterized complexity study of $\VRP$ and its variants. First, based on the different variants of vehicle routing found in the literature~\cite{braekers2016vehicleSurvey,mor2022vehicleSurvey}, we define a general  vehicle routing setting $\GVRS$ that provides a unified interface to the individual problems. We then focus on two important special cases: the $\VRP$ problem, as defined above, and its capacitated variant, where each route is additionally subject to a size constraint. Specifically, in $\LoadCVRP$, the input also contains the parameter $\ell$, and each route in the solution can only serve at most $\ell$ clients. This is motivated by real-life constraints, such as the limited cargo space of vehicles and the limited working hours of drivers.
On the other hand, the capacity constraint allows to encapsulate even more fundamental theoretical problems within the vehicle routing framework, such as \textsc{Triangle Packing}, corresponding to $\ell = 3$, and \textsc{$\ell$-Cycle Packing}. Similarly, we define $\GasCVRP$, where there is a limit on the edge-weight of each route, and $\LoadGasCVRP$, where both edge-weight and number of clients are limited.
The problem classification is covered in detail in \Cref{sec:classification}.

We then proceed to investigate the parameterized complexity of these problems with respect to the parameters such as the number of vehicles $k$, the number of depots $|D|$, the number of clients $|C|$, the total weight of the solution $r$, the treewidth of the graph $\tw$, and the maximum load $\ell$ in case of $\LoadCVRP$. Our findings are summarized in \Cref{tab:capacitated_overview_parameters}. We start in \Cref{sec:uncapacitated} with the uncapacitated case, i.e., the $\VRP$ problem. As this problem presents a generalization of $\TSP$, we can rule out tractability even for a constant number of vehicles and depots in \Cref{thm:vrp-d-k-paranp-hard}. In \Cref{thm:vrp-c-fpt} we also quickly conclude that the number of clients is a sufficient parameter for an $\FPT$ algorithm.
This result also directly implies that deciding whether a routing of weight at most $r$ exists can be done in $\FPT$-time in $r$.

\begin{table}[ht]
    \centering
    \begin{tabular}{llll} \toprule
        VRP Variant & Parameter                     & Complexity                     & Reference                                                      \\ \midrule
        $\VRP$
                    & $\cardinality D + k$          & $\paraNP$-hard                 & \Cref{thm:vrp-d-k-paranp-hard} ($\MetricTSP$)                  \\
                    & $\cardinality C$              & $\FPT$                         & \Cref{thm:vrp-c-fpt}                                           \\
                    & $r$                           & $\FPT$                         & \Cref{thm:vrp-r-fpt} ($r \ge \cardinality C$)                  \\
                    & $\tw$                         & $\FPT$                         & \Cref{thm:uncapacitated-tw-fpt}                                \\ \midrule
        $\LoadCVRP$
                    & $\ell$                        & $\paraNP$-hard                 & \Cref{thm:cvrp-cap-paraNP} ($\TrianglePacking$)                \\
                    & $\cardinality C$              & $\FPT$                         & \Cref{thm:cvrp-c-fpt}                                          \\
%                    & $r$                           & $\W[1]$-hard in general         & $\BinPacking$ (\Cref{thm:load-cvrp-w1-r})                      \\
                    & $r$                           & $\FPT$ if $0$-edges disallowed & \Cref{thm:load-cvrp-no-zero-r-fpt} ($r \ge \cardinality C$)    \\
                    & $k + \ell$                    & $\FPT$                         & \Cref{thm:cvrp-cap-k-fpt} ($k \cdot \ell \geq \cardinality C$) \\
                    & $\tw + \cardinality D$        & $\paraNP$-hard                 & \Cref{thm:cvrp-k-w1-on-trees} ($\BinPacking$)                  \\
                    & $\tw + k + \cardinality D$    & $\W[1]$-hard                    & \Cref{thm:cvrp-k-w1-on-trees} ($\BinPacking$)                  \\
                    & $\tw + \ell$                  & $\XP$                          & \Cref{thm:loadcvrp-xp}                                         \\ \midrule
        $\LoadGasCVRP$
                    & $\tw + \ell + \cardinality D$ & $\paraNP$-hard                 & \Cref{thm:loadgascvrp_np_const_demand} ($\NTDM$)               \\
                    & $\tw + g + \cardinality D$    & $\paraNP$-hard                 & \Cref{cor:loadgascvrp_np_const_gas} ($\NTDM$)                                   \\
       \bottomrule
    \end{tabular}
    \vspace{0.5em}
    \caption{Overview of our results for Vehicle Routing variants. We omit the results for $\GasCVRP$ and $\LoadGasCVRP$ that are analogous to the listed results for $\LoadCVRP$.}
    \label{tab:capacitated_overview_parameters}
\end{table}

The main result of \Cref{sec:uncapacitated} deals with the structural properties of our network, namely its treewidth.
While there are many known $\FPT$ algorithms parameterized by treewidth for various graph problems, the $\VRP$ problem presents a unique combination of several challenging aspects. The solution is composed of arbitrarily many individual objects, i.e., walks, each of which can also have arbitrary length.
Moreover, each walk may use the same edges and vertices multiple times. Neither the number of walks in the solution, nor the length of the individual walk, nor the multiplicity of an edge or a vertex in a walk is bounded by the parameter.

The closest known starting point on the approach to $\VRP$ is the $\FPT$ algorithm by Schierreich and Such\'{y}~\cite{schierreich22} for the \textsc{Waypoint Routing Problem} parameterized by treewidth, which is effectively a single-vehicle case of the $\VRP$. While resolving some of the challenges outlined above, their algorithm does not readily allow for the incorporation of the partitioning aspect as well, i.e., to deal with multiple vehicles, over which the clients need to be partitioned.
We generalize this approach further and obtain an $\FPT$ algorithm for $\VRP$ parameterized by treewidth (\Cref{thm:uncapacitated-tw-fpt}). In fact, the result we obtain holds for the more general Edge-Capacitated Vehicle Routing Problem ($\EVRP$), which allows additionally to specify arbitrary capacities on the edges, so that each edge can be used by the solution at most the specified number of times. One of the key ingredients behind the algorithm is the combinatorial observation that allows us to restrict the number of times an edge can be used by the solution. Based on this, we introduce an equivalent reformulation of $\VRP$ that only considers what we call \emph{valid routings}, inducing Eulerian submultigraphs. This reformulation can, in turn, be solved efficiently via dynamic programming over the tree decomposition. Unfortunately, the operations in dynamic programming required for dealing with multiple walks extend beyond the established single-exponential-time machinery developed by Bodlaender et al.~\cite{bodlaender2015deterministic}. The running time that our algorithm achieves is thus of the form $2^{\O(\tw \log \tw)} \cdot n^{\O(1)}$; improving this to purely single-exponential running time remains a challenging open question. Nevertheless, this result completes the classification with respect to the outlined parameters.

We then focus on capacitated vehicle routing in \Cref{sec:capacitated}.
Recall that $\LoadCVRP$, $\GasCVRP$, and $\LoadGasCVRP$ result from adding to $\VRP$ the restriction on the maximum number of clients per vehicle $\ell$, the maximum length of a tour $g$, or both restrictions simultaneously. As our work shows, these three versions largely behave in the same way with respect to the parameters in question; therefore, we mainly focus on $\LoadCVRP$ in order to introduce our results.
Starting from our results for the uncapacitated problem, we first show that the $\FPT$ algorithm for the number of clients and, in the case of $\GasCVRP$, the combined length of all routes, can be extended to work with these new constraints in \Cref{thm:cvrp-c-fpt} and \Cref{thm:gas-cvrp-r-fpt}.
In \Cref{thm:cvrp-cap-k-fpt}, we use this to derive further tractability results that involve the new parameter restricting maximum load.

The additional constraints allow us to model different packing problems as capacitated vehicle routing problems, and by showing such inclusions, we derive several intractability results.
In \Cref{sec:triangle-packing}, we show that $\LoadCVRP$ is \NP-hard even for maximum load $3$, by providing a reduction from $\TrianglePacking$. More surprisingly, by reduction from $\BinPacking$, we can also show that treewidth alone is not enough for fixed-parameter tractability in the presence of the capacity constraints, even when combined with the number of vehicles and the number of depot vertices.

Based on this hardness result, we move on to consider treewidth in combination with the capacity parameters.
We show complete intractability for $\LoadGasCVRP$, even when combining treewidth, the maximum load, and the number of depots as a parameter, using $\NTDM$ in \Cref{sec:NTDM}.
However, $\XP$ tractability can be achieved by combining treewidth and the capacity parameter in $\LoadCVRP$ and $\GasCVRP$, as well as combining both capacity parameters with treewidth in $\LoadGasCVRP$. These algorithms are the main results of \Cref{sec:capacitated}.
To obtain these, we first introduce a generalization of the $\BinPacking$ problem that acts as an interface to the capacitated vehicle routing problems, and prove tractability of this problem for different combinations of parameters.
Finally, in \Cref{sec:cvrp-xp} we present a dynamic programming algorithm that solves $\LoadGasCVRP$ by modeling the optimal combination of partial solutions as instances of generalized $\BinPacking$.

\section{Preliminaries}

Let $\N = \set{0, 1, 2, \dots}$ be the set of natural numbers, including zero.
For $n \in \N$, we abbreviate $[n] = \set{1, 2, \dots, n}$ and $[n]_0 = \set 0 \cup [n]$.
A \emph{multiset} is an unordered collection of elements of some universe $U$, where each element can occur multiple times.
Formally, we model multisets as functions mapping each possible element to the number of occurrences.
Alternatively, we use set notations with modified brackets as in these examples: $\multiset{1, 2, 2, 3}$ (the multiset containing $1$ and $3$ once, and $2$ twice), $\multiset{\floor{n / 2} \with n \in \N}$ (the multiset containing each natural number twice).
The cardinality of a multiset $A$ is denoted as $\cardinality A = \sum_{u \in U} A(u)$.
For some sub-universe $U' \subseteq U$, we denote the restriction of $A$ to $U'$ as $A_{\downarrow U'}$.
For two multisets $A, B$ over the same universe $U$, we denote by $A + B$ the multiset union defined by $(A + B)(u) = A(u) + B(u)$ for every $u \in U$.
For $n \in \N$, we denote by $n \cdot A$ the multiset obtained by uniting $n$ copies of $A$.
Let $G$ be an undirected (multi-)graph.
By convention, we usually abbreviate $n = \cardinality{V(G)}$.
A \emph{walk} $p$ in a graph $G$ is defined as a sequence of vertices and edges $(v_1, e_1, v_2, e_2, \dots, e_{\ell-1}, v_\ell)$,  
where the sequence begins at a vertex and traverses edges of the graph, allowing for the repetition of both vertices and edges.  
The \emph{length} of a walk is given by the number of edges it contains, counted with multiplicity.  
For simplicity, we may represent a walk of length $\ell - 1$ by the sequence of its vertices, i.e., $(v_1, v_2, \dots, v_\ell)$.
A walk is \emph{closed} if the first and last visited vertices are equal, otherwise it is \emph{open}.

A (multi-)graph $G$ is called \emph{Eulerian} if all vertices in $G$ have even degree.
Equivalently, $G$ is Eulerian when it is the (multiset-)union of closed walks.
Note that, unlike the classic (connected) notion of Eulerian graphs, this definition does not require 
$G$ to be connected.
Given an edge $e \in E(G)$ and number $k \in \N$, we denote by $G - k \cdot e$ the multigraph $G$ with $k$ occurrences of $e$ removed.

Let $w \colon E(G) \to \N$ be a weight function for the edges of $G$.
For a walk $p = (v_1, v_2, \dots, v_\ell)$ in $G$ we define the weight of $p$ as $w(p) = \sum_{i = 1}^{\ell - 1} w(v_i, v_{i+1})$.
Given vertices $u, v \in V(G)$, we denote by $w(u \leadsto v)$ the weight of the shortest path from $u$ to $v$.
For a submultigraph $H \subseteq G$ we define the weight of $H$ as $w(H) = \sum_{e \in E(H)} w(e)$, where parallel edges in $H$ should appear separately in the sum, that is, the weight of $H$ accounts for the multiplicity of the edges.

Finally, we recall the definitions of tree decompositions, treewidth, and nice tree decompositions, which we will use for our treewidth-based algorithms. 
\begin{definition}[Tree Decomposition, \cite{robertson1984graph, cygan2015parameterized}]
    A \emph{tree decomposition} of an undirected graph $G$ is a pair $\T = (T, \set{X_t}_{t \in V(T)})$ consisting of a tree $T$ and for each vertex $t$ of $T$ a set of vertices $X_t \subseteq V(G)$, so that
    \begin{itemize}
        \item $\bigcup_{t \in V(T)} X_t = V(G)$,
        \item for all $\set{u, v} \in E(G)$, there is $t \in V(T)$ such that $u, v \in X_t$, and
        \item for all $v \in V(G)$, the vertices $t \in V(T)$ with $v \in X_t$ form a connected subtree of $T$.
    \end{itemize}
    The sets $X_t$ are called \emph{bags}.
    We abbreviate $t \in V(T)$ as $t \in \T$.
    The \emph{width} of $\T$ is defined as $\opwidth(\T) = \max_{t \in \T} \cardinality{X_t} - 1$ and the \emph{treewidth} of $G$ is the minimum width of any tree decomposition of $G$.
\end{definition}

\begin{definition}[Nice Tree Decomposition, \cite{cygan2015parameterized}]
    A tree decomposition $\T = (T, \set{X_t}_{t \in V(T)})$ is called \emph{nice} if $T$ is rooted at a vertex $r \in V(T)$ with $X_r = \emptyset$ and all vertices $t \in V(T)$ are of one of the following kinds:
    \begin{itemize}
        \item Leaf Node: $t$ has no children and $X_t = \emptyset$.
        \item Introduce Vertex Node: $t$ has exactly one child $t'$ and there is a vertex $v \notin X_{t'}$ such that $X_t = X_{t'} \cup \set v$.
        \item Introduce Edge Node: $t$ has exactly one child $t'$, $X_t = X_{t'}$, and $t$ is labeled with an edge $\set{u, v} \in E(G)$ with $u, v \in X_t$.
        \item Forget Node: $t$ has exactly one child $t'$ and there is a vertex $v \in X_{t'}$ such that $X_t = X_{t'} \setminus \set v$.
        \item Join Node: $t$ has exactly two children $t_1, t_2$ and $X_t = X_{t_1} = X_{t_2}$.
    \end{itemize}
    For each edge $e \in E(G)$, there must be exactly one Introduce Edge Node labeled with $e$.
    If $G$ is a multigraph, each copy of an edge is introduced separately.

    For $t \in \T$, let $S$ be the set of all descendants of $t$ in $T$ (including $t$ itself).
    We denote the set of all introduced vertices up until $t$ as $V_t^\downarrow = \bigcup_{t' \in S} X_{t'}$.
    Similarly, we define $E_t^\downarrow$ to be the (multi-)set of all introduced edges up until $t$ and $G_t^\downarrow = (V_t^\downarrow, E_t^\downarrow)$.
\end{definition}

Recall that a tree decomposition always implies a nice tree decomposition of similar size.
\begin{theorem}[\cite{cygan2015parameterized}]
    \label{thm:nice-tree-decomposition-transform}
    There is an algorithm, that given a tree decomposition $\T$, in time $2^{\O(\opwidth(\T))}n$ transforms it into a nice tree decomposition of width $\O(\opwidth(\T))$ with\\ $\O(\opwidth(\T)^{\O(1)}n)$ nodes.
\end{theorem}

While the problem of determining the treewidth of a graph is $\NP$-hard~\cite{arnborg1987complexity}, there are $\FPT$ algorithms that compute exact or approximate tree decompositions.
For this work, we will use the following algorithm.

\begin{theorem}[\cite{korhonen2023single}]
    \label{thm:tw-2-approx}
    There is an algorithm that, given an $n$-vertex graph $G$ and an integer $k$, in time $2^{\O(k)}n$ either outputs a tree decomposition of $G$ of width at most $2k+1$ or concludes that $\tw(G) > k$.
\end{theorem}

\subsection{Classifying Vehicle Routing Problems}
\label{sec:classification}

The literature deals with many variants of vehicle routing problems~\cite{braekers2016vehicleSurvey,mor2022vehicleSurvey}.
To incorporate those in a single notion, we introduce the following \emph{Generalized Vehicle Routing Setting}.

\begin{definition}[Generalized Vehicle Routing Setting]
    \label{def:gvrs}
    % An instance of problem described by the \emph{Generalized Vehicle Routing Setting} ($\GVRS$) is defined by
    % \begin{enumerate}
    %         \item a directed graph $G$,
    %         \item edge weights $w\colon E(G) \to \N$,
    %         \item a nonempty set of \emph{depots} $D \subseteq V(G)$,
    %         \item a set of \emph{clients} $C \subseteq V(G)$,
    %         \item a \emph{number of vehicles} $k \in \N$, and
    %         \item \emph{additional constraints} modeled by an arbitrary predicate $\Gamma$.
    % \end{enumerate}
   %%%%%%%%%%%%%%%%%%%%%
   An instance of the \emph{Generalized Vehicle Routing Setting} ($\GVRS$) consists of a graph $G$ with edge weights $w\colon E(G) \to \N$, a nonempty set of \emph{depots} $D \subseteq V(G)$, a set of \emph{clients} $C \subseteq V(G)$, an integer $k \in \N$ representing the number of vehicles, and a predicate $\Gamma$ modeling additional constraints. 
   %%%%%%%%%%%%%%%%%
    We call a set of $k'$ walks $R = \set{p_1, p_2, \dots, p_{k'}}$ together with an assignment of vehicles $A\colon C \to [k']$ a \emph{routing} for $(G, w, D, C, k, \Gamma)$ if
    \begin{enumerate}
        \item $k' \le k$ (the routing uses at most $k$ vehicles),
        \item $\forall p = (v_1, v_2, \dots, v_\ell) \in R: v_1 \in D$ (every walk starts in one of the depots),
        \item $\forall c \in C: c \in p_{A(c)}$ (every client is visited by its assigned vehicle), and
        \item $\Gamma(R, A)$ (the additional constraints are satisfied).
    \end{enumerate}
    The weight of a routing $R$ is defined as
    $ w(R) = \sum_{p \in R} w(p)$.
    
    A problem instance may ask whether a feasible routing exists or whether there exists a routing $R$ with $w(R) \le r$,  for some $r \in \N$.
    A routing of minimum weight is called \emph{optimal}.
\end{definition}

We can see that vehicle routing problems generalize several well-studied combinatorial problems using this framework.
The \emph{Traveling Salesperson Problem} ($\TSP$) is equivalent to a $\GVRS$ where $k = 1$, $C = V(G)$, and with an additional constraint of
\[
    \Gamma_{\text{simple-cycle}}(R, A) \Iff \forall p \in R: p \text{ is a simple cycle}.
\]
The additional constraints can also require a strict form of routing, such as packing problems.
The \emph{$P_3$-packing} problem is equivalent to a $\GVRS$ with $\cardinality{V(G)} = 3k$, $C = D = V(G)$, and the additional constraint
\[
    \Gamma_{P_3}(R, A) \Iff \forall p \in R: p \text{ is a simple path of 3 vertices}.
\]

\subsection{Studied Problems}

This paper studies four vehicle routing problems, denoted here as $\VRP$, $\LoadCVRP$, $\GasCVRP$, and $\LoadGasCVRP$ for undirected graphs.
% As we will study all of them in terms of their complexity with regard to treewidth, we will only consider the case where $G$ is undirected and $w$ is symmetric.
The problem definitions will contain the additional constraint that each vehicle must return to the depot at which it started.
More formally, we will use
\[
    \Gamma_{\text{closed}}(R, A) \Iff \forall p \in R: p \text{ is closed}.
\]

\begin{definition}[$\VRP$]
    \label{def:vrp}
    The \emph{(Uncapacitated) Vehicle Routing Problem} ($\VRP$) is a $\GVRS$ with an additional parameter $r \in \N$.
    It asks whether there exists a routing $(R, A)$ with $w(R) \le r$, subject to $\Gamma = \Gamma_{\text{closed}}$.
\end{definition}

% When depicting $\VRP$ instances, we will mark depots by square vertices and clients by a double border, see for example \Cref{fig:vrp-example}. 

% \begin{figure}[ht]
%     \centering
%     \tikz [simple necklace layout, necklace routing]
%         \graph [VRP instance] {
%             "$d_1$"[depot] -- v1/ -- v2/ -- "$c_1$"[client],
%             v3/ -- "$c_2$"[client] -- "$d_2$"[depot, client],
%             v3/ -- {v4/, "$d_1$"},
%             "$d_1$" -- v4/ -- "$d_2$",
%             "$d_1$" -- "$c_1$",
%         };
%     \caption{An example $\VRP$ instance with $D = \set{d_1, d_2}$ and $C = \set{c_1, c_2, d_2}$.}
%     \label{fig:vrp-example}
% \end{figure}

In \Cref{sec:uncapacitated}, we obtain the main result for $\VRP$, a parameterized algorithm for instances with bounded treewidth, by considering $\EVRP$, a slightly generalized version of $\VRP$.

\begin{upright-definition}[$\EVRP$]
    \label{def:evrp}
    The \emph{Edge-Capacitated Vehicle Routing Problem} ($\EVRP$) is a $\GVRS$ with additional parameters $r \in \N$ and $\kappa\colon E(G) \to \N$.
    It asks whether there exists a routing $(R, A)$ with $w(R) \le r$, subject to $\Gamma = \Gamma_{\text{closed}} \cap \Gamma_\kappa$, where
    \[
        \Gamma_\kappa(R, A) \Iff \forall e \in E: e \text{ is used at most $\kappa(e)$ times in $R$}. \lipicsEnd
    \]
\end{upright-definition}

The literature on vehicle routing problems often considers \emph{capacitated} variants~\cite{braekers2016vehicleSurvey,mor2022vehicleSurvey}, usually denoted as $\CVRP$.
In contrast to $\EVRP$, these capacities are placed on the vehicles individually.
However, the term $\CVRP$ can refer to different kinds of capacities: \emph{load} and \emph{gas} capacities.
In this paper, we consider both interpretations of $\CVRP$.

In $\LoadCVRP$, each client $c$ has a demand for $\Lambda(c)$ units of some good, which must be delivered by its assigned vehicle.
The vehicles, in turn, have a capacity $\ell$ on the \emph{load} they can carry on their trip.
Notably, this capacity $\ell$ is equal for all vehicles.

\begin{definition}[$\LoadCVRP$]
    The \emph{Load-Capacitated Vehicle Routing Problem} ($\LoadCVRP$) is a $\GVRS$ with additional parameters $r \in \N$, $\ell \in \N$, and $\Lambda\colon C \to \N^+$.
    It asks whether there exists a routing $(R, A)$ with $w(R) \le r$, subject to $\Gamma = \Gamma_{\text{closed}} \land \Gamma_{\ell, \Lambda}$, where
    \[
        \Gamma_{\ell, \Lambda}(R, A) \Iff \forall i \in [\cardinality R]: \sum_{c \in A^{-1}(i)} \Lambda(c) \le \ell.
    \]
\end{definition}

A similar restriction is given in $\GasCVRP$, where each route is limited by the amount $g$ of \emph{gas} (fuel) a vehicle can carry.
Similar to $\LoadCVRP$, the capacity $g$ is equal for all vehicles.

\begin{upright-definition}[$\GasCVRP$]
    The \emph{Gas-Capacitated Vehicle Routing Problem} ($\GasCVRP$) is a $\GVRS$ with additional parameters $r \in \N$ and $g \in \N$.
    It asks whether there exists a routing $(R, A)$ with $w(R) \le r$, subject to $\Gamma = \Gamma_{\text{closed}} \land \Gamma_g$, where
    \[
        \Gamma_g(R, A) \Iff \forall p \in R: w(p) \le g. \lipicsEnd
    \]
\end{upright-definition}

Finally, $\LoadGasCVRP$ combines the constraints of $\LoadCVRP$ and $\GasCVRP$.

\begin{definition}[$\LoadGasCVRP$]
    The \emph{Load-and-Gas-Capacitated Vehicle Routing Problem} ($\LoadGasCVRP$) is a $\GVRS$ with additional parameters $r \in \N$, $\ell \in \N$, $\Lambda\colon C \to \N^+$, and $g \in \N$.
    It asks whether there exists a routing $(R, A)$ with $w(R) \le r$, subject to $\Gamma = \Gamma_{\text{closed}} \land \Gamma_{\ell, \Lambda} \land \Gamma_g$.
\end{definition}

\section{Uncapacitated Vehicle Routing}
\label{sec:uncapacitated}

In this section, we establish the parameterized landscape around the Uncapacitated Vehicle Routing Problem.
We show hardness of $\VRP$ when parameterized by the number of depots and vehicles.
In turn, we show that $\VRP$ is tractable when parameterized by the number of clients.
Finally, we present an algorithm which proves that $\EVRP$ (and thereby $\VRP$) is tractable on graphs of bounded treewidth.

Note that for $\VRP$ routings, it suffices to find a set $R$ of at most $k$ closed walks that connect every client to a depot.
This is because the predicate $\Gamma$ does not pose any requirements on the assignment of vehicles $A$.
In fact, an assignment $A$ can be chosen based on $R$ by arbitrarily selecting one of the vehicles which visit a given client.
We thus do not mention the assignment $A$ any further in this section.

%\subsection{Tractability on general graphs}

%We start by formalizing the connection between $\TSP$ and $\VRP$.
Note that $\MetricTSP$ is a special case of \VRP~ where  $C = V(G)$, $D = \set v$, and $k = 1$.  This yields the following result:
\begin{theorem}
    \label{thm:vrp-d-k-paranp-hard}
    $\VRP$ parameterized by $\cardinality D + k$ is $\paraNP$-hard.
\end{theorem}
% \begin{proof}
%     This follows from the $\NP$-hardness of $\MetricTSP$.
%     Given $G$ and metric weights $w\colon E(G) \to \N$, choose an arbitrary vertex $v \in V(G)$ as the start of the tour and set $C = V(G)$, $D = \set v$, and $k = 1$.
%     Any $\TSP$-tour corresponds to a routing of the same weight.
%     Conversely, any routing can be shortcut to form a $\TSP$-tour, which has at most the same weight due to $w$ being a metric.
% \end{proof}

Note that the above observation still requires a large number of clients.
In the next theorem, we see that we can efficiently solve $\VRP$ instances for constant numbers of clients.

\begin{theorem}
    \label{thm:vrp-c-fpt}
    There is an algorithm that computes an optimal $\VRP$ routing in $\cardinality{C}^{\O(\cardinality C)} \cdot n^{\O(1)}$ steps.
\end{theorem}
\begin{proof}
    We iterate over all partitions $\P$ and permutations $\sigma$ of $C$.
    If $\cardinality\P > k$, we skip the iteration.
    Otherwise, for every block $X \in \P$, let $c_1, c_2, \dots, c_\ell$ be the order of $X$ according to $\sigma$.
    We choose the depot $d_X \in D$ such that $w_X = w(d_X \leadsto c_1 \leadsto c_2 \leadsto \cdots \leadsto c_\ell \leadsto d_X)$ is minimized.
    This gives a routing that uses at most $k$ vehicles with weight $\sum_{X \in \P} w_X$.

    There are $\cardinality{C}^{\cardinality{C}} \cdot \cardinality{C}! \le \cardinality{C}^{\O(\cardinality C)}$ iterations.
    By first precomputing a distance-table over $D \cup C$ via any all-pairs-shortest-paths algorithm in $n^{\O(1)}$, each iteration takes
    \begin{itemize}
        \item $\O(\cardinality C)$ time to compute $w(c_1 \leadsto c_2 \leadsto \cdots \leadsto c_\ell)$, and
        \item $\O(\cardinality D)$ time to compute the weight $w(d \leadsto c_1) + w(c_\ell \leadsto d)$ for every depot $d \in D$.
    \end{itemize}
    In total, this gives a running time of $\cardinality{C}^{\O(\cardinality C)} (\cardinality C + \cardinality D) + n^{\O(1)} \le \cardinality{C}^{\O(\cardinality C)} \cdot n^{\O(1)}$.
    As any optimal routing $R$ induces a partition $\P$ and permutation $\sigma$, the found routing has minimum weight.
\end{proof}

While the above algorithm outputs an optimal routing, it can also be used to solve the decision variant of $\VRP$ as defined in \Cref{def:vrp}.
It also yields the following corollary:

\begin{corollary}
    \label{thm:vrp-r-fpt}
    There is an algorithm that, given $r \in \N$, decides whether there exists a $\VRP$ routing of weight at most $r$ in $f(r) \cdot n^{\O(1)}$ steps, for some computable function $f$.
\end{corollary}
\begin{proof}
    We first contract any zero-weight edges in $G$.
    If the vertices $u, v$ are contracted, the new vertex $v'$ should be a client (depot) if either of $u$ and $v$ was a client (depot).
    This does not alter the existence of the routings we are interested in, as any routing which visits a new vertex $v'$ can visit all original vertices without incurring an additional cost via the contracted edge.
    Recall that we assume the weights $w$ to be integral.
    Now that every edge has positive weight, a routing $R$ can contain at most $w(R)$ vertices.
    Thus, we can reject the instance if $r < \cardinality C$, and otherwise use the algorithm from \Cref{thm:vrp-c-fpt} to obtain an optimal routing and compare its weight with $r$.
\end{proof}

%\subsection{Tractability on graphs of bounded treewidth}

Next, we present an FPT algorithm for $\VRP$ on graphs of bounded treewidth parameterized by the treewidth. Our algorithm generalizes the prior work by Schierreich and Su\'{c}hy~\cite{schierreich22} on the \emph{Waypoint Routing Problem} ($\WRP$). $\WRP$ can be formulated in terms of $\GVRS$ as follows:

\begin{upright-definition}[$\WRP$]
    The \emph{Waypoint Routing Problem} ($\WRP$) is a $\GVRS$ with $k = 1$, $D = \set s$, a designated vertex $t \in V(G)$, and additional parameters $r \in \N$ and $\kappa\colon E(G) \to \N$.
    It asks whether there exists a routing $R$ (consisting of a single walk, due to $k = 1$) with $w(R) \le r$, subject to $\Gamma = \Gamma_t \cap \Gamma_\kappa$, where
    \[
        \begin{aligned}
            \Gamma_t(R, A) &\Iff \forall p \in R: p \text{ ends at $t$,} \\
            \Gamma_\kappa(R, A) &\Iff \forall e \in E: e \text{ is used at most $\kappa(e)$ times in $R$}.
        \end{aligned}\lipicsEnd
    \]
\end{upright-definition}

\begin{theorem}[\cite{schierreich22}]
    \label{thm:wrp-tw-fpt}
    There is an algorithm that computes an optimal $\WRP$ routing in $2^{\O(\tw)} \cdot n^{\O(1)}$ steps, where $\tw$ denotes the treewidth of $G$.
\end{theorem}

Schierreich and Su\'{c}hy~\cite{schierreich22} obtain this result by reducing $\WRP$ to a normal form where $s = t$, effectively replacing the constraint $\Gamma_t$ by $\Gamma_{\text{closed}}$.
Next, they reformulate the problem to remove the parameter $\kappa$:
Every edge $e \in E(G)$ is replaced by $\kappa(e)$ many parallel edges $e^1, e^2, \dots, e^{\kappa(e)}$, each with weight $w(e)$.
Given the obtained multigraph $G'$, $\WRP$ now asks to find a connected, Eulerian subgraph $H \subseteq G$, which contains the depot $s$ and all clients $C$.

Since $\WRP$ is very similar to a single-vehicle case of $\VRP$, we aim to adapt the techniques from \cite{schierreich22} to our work.
In particular, we slightly generalize the $\VRP$ problem and consider $\EVRP$ that incorporates edge capacities (see \Cref{def:evrp}), making $\WRP$ a proper special case of our problem.
The introduction of edge capacities $\kappa$ to $\VRP$ yields the same reformulation to multigraphs as found for $\WRP$.
\begin{lemma}
    \label{thm:evrp-routing-multigraph-equiv}
    Let $G'$ be the multigraph obtained by replacing all edges $e$ of $G$ by $\kappa(e)$ parallel edges as described above.
    There is an $\EVRP$ routing $R$ of weight $r \in \N$ in $G$ if and only if there is an Eulerian subgraph $H \subseteq G'$ of weight $w(H) = r$ with $ C \subseteq V(H)$ and at most $k$ connected components, so that every client $c \in C$ is reachable from a depot $d \in D$.
\iffalse
    \textcolor{blue}{   Let $G'$ be the multigraph obtained by replacing all edges $e$ of $G$ by $\kappa(e)$ parallel edges as described above.
    There is an $\EVRP$ routing $R$ of weight $r \in \N$ in $G$ if and only if there is a subgraph $H \subseteq G'$ of weight $w(H) = r$ with $V(H) \subseteq C$ and at most $k$ connected components such that each connected component is Eulerian and every client $c \in C$ is reachable from a depot $d \in D$.}
    \fi
\end{lemma}
\begin{proof}
    We proceed similar to \cite{schierreich22}.
    Given an $\EVRP$ routing $R$, we can transform the walks of $R$ into a set of walks $R'$ in $G'$ by replacing the first occurrence of an edge $e$ in $R$ by $e^1$, the second by $e^2$, and so on.
    As $R$ is an $\EVRP$ routing, it uses each edge $e$ at most $\kappa(e)$ times.
    Define $H$ so that $V(H)$ and $E(H)$ are the vertices and edges traversed by the walks of $R'$ respectively.
    Clearly, $w(H) = w(R') = w(R)$.
    As all walks in $R'$ are cyclic (Eulerian), $H$ is the required subgraph.
    Additionally, as $R$ is a routing, every client must be reachable from a depot in $H$.
    Finally, $\cardinality R \le k$ implies that $H$ has at most $k$ connected components.

    In turn, given $H$, we can obtain an $\EVRP$ routing by finding an Eulerian trail in every connected component of $H$ and mapping all edges $e^i$ from $G'$ to the edge $e$ in $G$.
    It is easy to check that this routing fulfills the conditions of the lemma.
\end{proof}

Based on \Cref{thm:evrp-routing-multigraph-equiv}, we reuse the term \emph{routing} to refer to Eulerian subgraphs $H \subseteq G'$ that connect every client to a depot and contain at most $k$ connected components.
Furthermore, we shift the goal of our algorithm to finding a suitable subgraph $H \subseteq G'$.
Before we can move on from $G$ and $\kappa$, we need to make one more simplification to the problem:
Although replicating all edges $e$ by $\kappa(e)$ copies greatly increases the size of the input graph, this increase can be avoided by the following observation, which can be proven easily.

\begin{lemma}
    \label{thm:evrp-remove-edge-twice}
    Let $H \subseteq G'$ be an $\EVRP$ routing and $e \in E(H)$ be an edge with multiplicity greater than $2$.
    Then $H - 2e$ is an $\EVRP$ routing and $w(H - 2e) \le w(H)$.
\end{lemma}
% \begin{proof}
%     Since $e$ occurs more than twice, removing two copies of $e$ does not disconnect the endpoints of $e$.
%     Additionally, $H - 2e$ is still Eulerian, because the degree of both endpoints of $e$ is decreased by $2$.
%     Finally, as $w(e) \ge 0$, the weight of $H - 2e$ cannot be greater than the weight $H$.
% \end{proof}

By repeatedly applying \Cref{thm:evrp-remove-edge-twice}, we see that it suffices to solve $\EVRP$ instances with $\kappa(e) \le 2$ for all edges $e$.
In this case, the introduction of $\kappa(e)$ copies of an edge $e$ can only increase the input size by a factor of $2$.
Therefore, we will not mention the simple graph $G$ and the edge capacities $\kappa$ anymore.
Instead, we work on multigraphs where each edge can only be used once.

It seems like the algorithm from \Cref{thm:wrp-tw-fpt} cannot be used directly to solve $\EVRP$.
 However, its general approach still applies to vehicle routing problems.
The main difference is that in $\WRP$, all partial solutions eventually merge to a connected subgraph containing the source vertex $s$, whereas in $\EVRP$, there are multiple connected components, each of which contains one of the depots.
It turns out that this observation directly translates to an $\FPT$-algorithm for $\EVRP$ parameterized by treewidth and the number of connected components.

\begin{theorem}
    \label{thm:evrp-tw-d-fpt}
    There is an algorithm that computes an optimal $\EVRP$ routing in $\binom{\cardinality D}{d^*} (\tw + d^*)^{\O(\tw + d^*)} \cdot n^{\O(1)}$ steps where $d^* = \min\set{\cardinality D, k}$.
\end{theorem}
\begin{proof}
    In the case that $\cardinality D \le k$, we use the dynamic programming algorithm from \cite{schierreich22} with only a slight modification.
    Instead of adding the source $s$ to every bag of the tree decomposition, all depots are added, increasing the size of the tree decomposition to $\tw + \cardinality D$.
    The solution is a minimum-weight partial solution compatible with $(D, \emptyset, D[\emptyset])$ at the root.

    If $\cardinality D > k$, we guess the set $D^*$ of depots that is used by an optimal solution and compute the optimal routing for every such $D^*$.
    As any routing can have at most $k$ connected components, there are $\binom{\cardinality D}{k}$ different sets $D^*$ to try.
\end{proof}

Importantly, the optimization employed by \cite{schierreich22} to improve the running time from $\tw^{\O(\tw)} \cdot n^{\O(1)}$ to $2^{\O(\tw)} \cdot n^{\O(1)}$ cannot be used to improve \Cref{thm:evrp-tw-d-fpt}.
The optimization is based on the $\opreduce$ operation from \cite{bodlaender2015deterministic} that merges partial solutions which are equivalent when considering that they still have to be extended to the signature $\set{\set s}$ of a full solution.
In \Cref{thm:evrp-tw-d-fpt}, the signature at the root node is read from the partition $D[\emptyset]$, not $\set D$.
Thus, a generalization of the setup by \cite{bodlaender2015deterministic} to more complex partitions is needed to achieve a similar speedup.

While \Cref{thm:evrp-tw-d-fpt} yields that $\EVRP$ parameterized by $\tw + \cardinality D$ is in $\FPT$, it seems that the information tracked in the DP by \cite{schierreich22} is almost sufficient for $\EVRP$.
As we will see in the remainder of this section, this intuition is correct.
By extending the weighted-partition setup of \cite{bodlaender2015deterministic} and adapting the DP by \cite{schierreich22}, we manage to develop an algorithm for $\EVRP$ which is $\FPT$ when parameterized by just the treewidth of $G$.

To adapt the setup we first introduce the coarsening relation and join operation on partitions.
After presenting our extension of weighted partitions and adjusting the operations on weighted partitions by \cite{bodlaender2015deterministic} to fit our definition, we can finally go on to present our tree decomposition based algorithm for \EVRP.

\subsection{Coarsening of Partitions}
%\todo[inline]{Each block is a set, so I changed to capital letters. [S: I agree]}
Let $U$ be a finite set and let $\parP, \parQ$ be partitionings of $U$.
We call the elements of a partition, $X \in \parP$, \emph{blocks}.
We say that $\parP$ is \emph{finer} than $\parQ$, denoted as $\parP \sqsubseteq \parQ$, if every block $X \in \parP$ is contained in some block $Y \in \parQ$, that is $X \subseteq Y$.
Conversely, we say that $\parQ$ is \emph{coarser} than $\parP$, or is a \emph{coarsening} of $\parP$.

\begin{example}
    Let $\parP = \set{\set{1, 2}, \set{3}, \set{4, 5, 6}}$ and $\parQ = \set{\set{1, 2, 3}, \set{4, 5, 6}}$ be partitionings of the universe $U = \set{1, 2, 3, 4, 5, 6}$.
    Since $\set{1, 2}, \set{3} \subseteq \set{1, 2, 3}$ and $\set{4, 5, 6} \subseteq \set{4, 5, 6}$, we have $\parP \sqsubseteq \parQ$.

    We can also see that $\set{U}$ is the coarsest partitioning while $\set{\set{1}, \set{2}, \set{3}, \set{4}, \set{5}, \set{6}}$ is the finest partition.
\end{example}

%\todo[inline]{I have removed this part: One can verify that the set of all partitions over a universe together with the coarsening relation forms a lattice.
%We denote the join operation of this lattice as $\sqcup$. [S: Agree]}
For two partitions $\parP, \parQ$ of $U$, we define the binary operation $\parP \sqcup \parQ$ to denote the finest common coarsening of $\parP$ and $\parQ$.
More formally, the finest common coarsening of $\parP$ and $\parQ$ is a partition $\parR$ such that:
\begin{enumerate}
    \item $\parP \sqsubseteq \parR$ and $\parQ \sqsubseteq \parR$
    \item There is no partition $\parR' \neq \parR$, satisfying 1.\ with $\parR' \sqsubseteq \parR$
\end{enumerate}
\begin{example}
    Let $\parP = \set{\set{1, 2}, \set{3, 4}, \set{5, 6}}$ and $\parQ = \set{\set{1}, \set{2, 3}, \set{4}, \set{5}, \set{6}}$ be partitions of $U = \set{1, 2, 3, 4, 5, 6}$.
    We have $\parP \sqcup \parQ = \set{\set{1, 2, 3, 4}, \set{5, 6}}$.

    One can verify that this is a coarsening of both $\parP$ and $\parQ$.
    Note that, $\set{1, 2, 3, 4}$ needs to be a part of any valid coarsening of both $\parP$ and $\parQ$.
    $\parP$ requires $1$ and $2$ to stay together, as well as $3$ and $4$, while $\parQ$ requires $3$ and $4$ to end up in the same block.
    Therefore this is the finest possible coarsening.
\end{example}

For $V \subseteq U$, we denote by $U[V] = \set V \cup \set{\set u \with u \in U \setminus V}$ the partition where all blocks are singletons, except for $V$, which is one block.
%\todo[inline]{$\parP_{\downarrow V} = \set{x \cup V \with x \in \parP} \setminus \set{\emptyset}$ is wrong. [S: Correct, thanks for pointing out.] I corrected it.}
We denote by $\parP_{\downarrow V} = \set{X \cap V \with X \in \parP} \setminus \set{\emptyset}$ the restriction of $\parP$ to $V$.
Finally, for $V \supseteq U$, we define $\parP_{\uparrow V} = \parP \cup \set{\set v \with v \in V \setminus U}$ as the extension of $\parP$ to $V$.

\subsection{Computation of the Dynamic Programming Array}

Our dynamic programming procedure computes a set of partial solutions conforming with a particular signature at each node of the tree decomposition.

\begin{upright-definition}
    \label{def:vrp-dp}
    For every $t \in \T$, $X \subseteq X_t$, $L \subseteq X$, and $c \le k$, with $C \cap X_t \subseteq X$, we call $S = (X, L, c)$ a presignature at $t$.
    A presignature together with a partition $\P$ of $X$ and a subset of depot-connected blocks $\mathcal{B} \subseteq \P$ is a solution signature at $t$.
    A subgraph $H \subseteq G_t^\downarrow$ is a partial solution compatible with the solution signature $(X, L, c, \P, \mathcal{B})$ at $t$, if
    \begin{enumerate}
        \item $C \cap V_t^\downarrow \subseteq V(H)$
        \item $V(H) \cap X_t = X$,
        \item a vertex $v \in V(H)$ has odd degree in $H$ iff $v \in L$,
        \item there are exactly $c$ connected components $A$ of $H$ with $V(A) \cap X_t = \emptyset$ and $V(A) \cap D \ne \emptyset$,
        \item every other connected component $A$ has $X_A \coloneq V(A) \cap X_t \in \P$ and $V(A) \cap D \ne \emptyset$ if and only if $X_A \in \mathcal{B}$. \lipicsEnd
    \end{enumerate}
\end{upright-definition}

%\todo[inline]{adding intuition. [S: nice]}
Before proceeding, we explain the intuition behind each required condition in the definition above.
Intuitively, all the visited clients so far should be in $H$ (1), as otherwise there will be an unserved client in the final solution.
\( X \) represents the set of vertices of \( H \) that are contained within the bag (2). The set \( L \) consists of the odd degree vertices of \( H \).
Note that due to the Eulerian nature of $H$, its odd-degree vertices can only be in the bag (3).
The parameter \( c \) denotes the number of connected components of \( H \) that lie entirely outside the bag and each contain at least one depot (4). Finally, \( \mathcal{P} \) indicates the connected components of \( H \), present in the bag, and \( \mathcal{B} \) is the set of those connected components of $H$ which contain at least one depot, in the bag (5).
%\todo[inline]{Until here}

\paragraph*{Marked partitions}
Let $w \in \N$, represent the weight of a partial solution $H$, we call a triple $(\P, \mathcal{B}, w)$ a marked partition.
Using dynamic programming, we compute a set $dp[t, X, L, c]$ of marked partitions $(\P, \mathcal{B}, w)$ for each $t \in \T$ and presignature $(X, L, c)$.
A marked partition $(\P, \mathcal{B}, w) \in dp[t, X, L, c]$ represents a partial solution compatible with the signature $(X, L, c, \P, \mathcal{B})$ at $t$, with weight $w$.

We further generalize certain operations originally introduced by Bodlaender et al.~\cite{bodlaender2015deterministic} and subsequently employed by Schierreich and Such\'{y}~\cite{schierreich22}, adapting them to our framework to serve as fundamental components in the dynamic programming state transition process.

For a family $\mathcal{S}$ of sets (over some universe $U$), we denote the set of all supersets of members of $\mathcal S$ by
\[
    \Sup(\mathcal{S}) = \left\{ S' \;\middle|\; \exists S \in \mathcal{S} \text{ such that } S \subseteq S' \right\}.
\]
In practice, the set $\Sup(\mathcal S)$ should not be constructed as it is only used to make the definitions below more readable.
Instead, an intersection $\P \cap \Sup(\mathcal S)$ should be computed by filtering $\P$ according to the definition of $\Sup(\mathcal S)$.

Let $ \mathcal{A} $ and $ \mathcal{C} $ be sets of marked partitions, where each partition in $ \mathcal{A} $ and $ \mathcal{C} $ partitions the universe $ U $. We define the following operations on sets of marked partitions:

\subparagraph*{Remove copies:}
Removes marked partitions where an equivalent marked partition with a smaller weight exists.
\[
    \oprmc(\mathcal{A}) \coloneq \set{ (\P, \mathcal{B}, w) \in \mathcal{A} \with \nexists(\P, \mathcal{B}, w')\in \mathcal{A}: w' < w }
\]

\subparagraph*{Union:}
Combines the marked partitions of $\A$ and $\mathcal{C}$ and discards duplicates.
\[
    \mathcal{A} \mincup \mathcal{C} \coloneq \oprmc(\mathcal{A} \cup \mathcal{C})
\]

\subparagraph*{Shift:}
Increases the weight of each marked partition in $\A$ by $w' \in \N$.
\[
    \opshift(w', \A) \coloneq \set{(\P, \mathcal{B}, w + w') \with (\P, \mathcal{B}, w) \in \A}
\]

\subparagraph*{Glue:}
This operation is applied when an edge $\{u, v\}$ is introduced to the bag.
Intuitively, the $\opglue$ operation merges the connected components containing $u$ and $v$ under the assumption that the edge $\{u, v\}$ is included in the solution $H$.  
If either connected components of $u$ or $v$ contains a depot,  
then the merged connected component is likewise marked as containing a depot.  
Additionally, the weight of the edge $\{u, v\}$ is added to the total weight of $H$.  

More formally, for $u, v \in U$, this operation merges the blocks containing $u$ and $v$ within the partition of each entry.  
The resulting merged blocks are marked as connected to a depot  
if at least one of the original blocks was marked as such prior to the merge.
We also add the abbreviation $\opglue_w$ defined as:
\begin{align*}
    \opglue(\set{u, v}, \A)   &\coloneq \oprmc(\set{(\P', \P' \cap \Sup(\mathcal{B}), w) \with (\P, \mathcal{B}, w) \in \A \land \P' = \P \sqcup U[\set{u, v}]}) \\
    \opglue_w(\set{u, v}, \A) &\coloneq \opshift(w(\set{u, v}), \opglue(\set{u, v}, \A))
\end{align*}

\subparagraph*{Insert:}
This operation is used when a vertex $v$ is introduced to the bag.
Intuitively, the $\opins$ operation adds $\{v\}$ to the connected components of $H$ in the previous bag, since it is an isolated vertex, as none of its neighboring edges have been introduced yet. In case $v$ is a depot, that is $v \in D$, it adds $\{v\}$ also to $\mathcal{B}$.  

Technically, for $V$ with $V \cap U = \emptyset$ and a set $D$, $\opins$ expands the universe by $V$ by inserting every $v \in V$ as a singleton block in every entry of $\A$.
The new blocks for $v \in V$ should be marked as connected to a depot if $v \in D$ and it is denoted by
\[
    \opins(V, D, \A) \coloneq \set{(\P_{\uparrow U \cup V}, \mathcal{B} \cup (D \cap V)[\emptyset], w) \with (\P, \mathcal{B}, w) \in \A}
\]

\subparagraph*{Project:}
This operation is used when a vertex $v$ is forgotten, and by forgetting $v$ no connected component is separated from $H$. So it first discards entries where $\{v\}$ was a connected component in the previous bag, as it can not be case. Then, among the remaining valid entries of the previous bag, it removes $v$ from the connected components. An illustration of this can be found in~\cref{fig:operations}.
More formally, for $V \subseteq U$, it removes $V$ from the universe and discard entries where an entire block of a partition is deleted:
\[
    \opproj(V, \A) \coloneq \oprmc\left( \set*{ (\P_{\downarrow U \setminus V}, \mathcal{B}_{\downarrow U \setminus V}, w ) \with \begin{aligned}
        & (\P, \mathcal{B}, w) \in \A, \\
        & \forall v \in V \exists v' \in U \setminus V: U[\set{v,v'}] \sqsubseteq \P
    \end{aligned}}\right)
\]

\subparagraph*{Detach:}
This operation is used also when a vertex $v$ is forgotten, with the difference that forgetting $v$ results in a connected component being completely detached from $H$. In this case, $v$ should be connected to a depot as otherwise there will be a connected component with no depot.
So $\opdetach$ only keeps the entries where $\{v\}$ (with no other vertex) was one connected component in the previous bag and also was connected to a depot. Then, among the above valid entries, it removes $v$ from the connected component. \Cref{fig:operations} provides a visual representation of the $\opdetach$ operation.
In precise terms, for $V \subseteq U$, this operation removes $V$ from the universe and discards entries containing blocks which are only partly deleted.
Particularly, we only keep entries where the blocks with elements from $V$ are completely deleted and are marked as connected to a depot. 
\[
    \opdetach(V, \mathcal{A}) \coloneq \oprmc\left(\set*{(\P_{\downarrow U \setminus V}, \mathcal{B}_{\downarrow U \setminus V}, w) \with
    \begin{aligned}
        & (\P, \mathcal{B}, w) \in \mathcal{A}, \\
        & \forall x \in \P: x \cap V \neq \emptyset \implies x \subseteq V \land x \in \mathcal{B}
    \end{aligned}}\right)
\]

\begin{figure}[ht]
  \centering
  \includegraphics[width=0.8\textwidth]{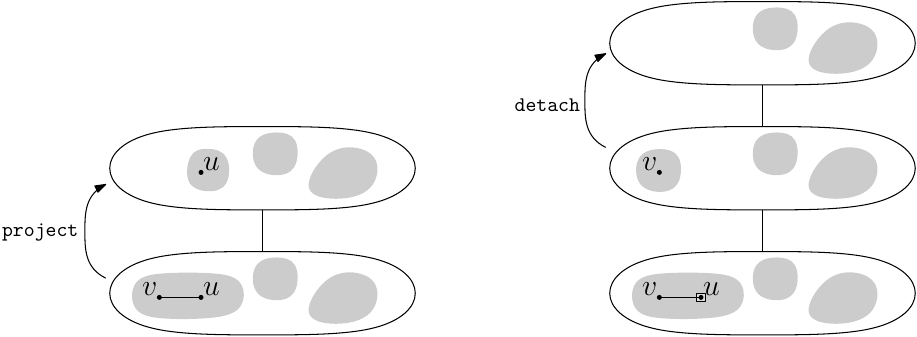}
  \caption{Illustration of $\opproj$ and $\opdetach$ operations, left and right respectively. Gray areas indicate the connected components, square indicates a depot, and vertex $v$ is being forgotten. Observe that on the left, no connected component is separated by forgetting $v$, while on the right, a connected component is completely detached by forgetting $v$. For $\opdetach$, it is required that $v$ is connected to a depot, which is showcased by the additional bag below.}
  %an extra bag is presented to just show that $v$ should have been connected to a depot (only entries with this property are among the valid ones).}
  \label{fig:operations}
\end{figure}

\subparagraph*{Join:}
Constructs all combinations of marked partitions from $\A$ and $\mathcal{C}$.
The blocks which share an element are combined, and the combined block is connected to a depot if any of the blocks was before the merge. It is denoted by
\[
    \opjoin(\A, \mathcal{C}) \coloneq \oprmc\left(\set*{(\P', \P' \cap \Sup(\mathcal{B}_1 \cup \mathcal{B}_2), w_1 + w_2) \with \begin{aligned}
        & (\P, \mathcal{B}_1, w_1) \in \A, \\
        & (\parQ, \mathcal{B}_2, w_2) \in \mathcal{C}, \\
        & \P' = \P \sqcup \parQ
    \end{aligned}}\right)
\]

%\todo[inline]{Do we need a proof for this?} 
Similarly to the observations of \cite{bodlaender2015deterministic,schierreich22}, it is easy to check that the following holds:
\begin{lemma}
    The operators described above, excluding $\opjoin$, can be executed in $s \cdot \cardinality{U}^{\O(1)}$ time, where $s$ is the size of the operation input.
    The $\opjoin$ operator can be executed in $\cardinality \A \cdot \cardinality{\mathcal{C}} \cdot \cardinality{U}^{\O(1)}$ time.
\end{lemma}

Given a nice tree decomposition of the input graph, we are ready to define the computation at every node  using the operations above. 

\paragraph*{Leaf Node}
In a leaf node $t$, we have $X_t = \emptyset$, so we set $dp[t, \emptyset, \emptyset, 0] = \set{(\emptyset, \emptyset, 0)}$.

\paragraph*{Introduce Vertex Node}
Suppose that node $t$ with child $t'$ introduces vertex $v$ and let $(X, L, c)$ be a presignature at $t$.
As none of the edges incident to $v$ have been introduced yet, $v$ cannot have odd degree in any partial solution at $t$ and $v$ can only be a singleton component.
We set
\[
    dp[t, X, L, c] = \begin{cases}
        \mathtt{ins}(\set v,D, dp[t', X \setminus \set v, L, c]) &\text{if } v \in X \cap \compl L, \\
        dp[t', X, L, c] &\text{if } v \in \compl X \cap \compl L, \\
        \emptyset &\text{otherwise}.
    \end{cases}
\]

\paragraph*{Introduce Edge Node}

Suppose that $t$ is an introduce edge node with child $t'$ and $E_t = E_{t'} + \multiset{\set{u,v}}$, let $(X, L, c)$ be a presignature at $t$.
We might use the edge, in which case the parity of $u,v$ in $L$ flips, or ignore it.
We define
\[
   dp[t,X,L,c] = \begin{cases} 
        dp[t', X, L, c]\mincup \mathtt{glue}_{w}(\set{u,v}, dp[t', X, L \symdiff \set{u,v}, c]) &\text{if } u,v \in X,\\
        dp[t', X, L, c] &\text{otherwise}. 
   \end{cases}
\]  

\paragraph*{Forget Node}
Suppose that $t$ is a node with child $t'$ forgetting vertex $v$ and let $(X, L, c)$ be a presignature at $t$.
We know that $v \notin X$ and thus consider all marked partitions from presignature $(X,L,c)$ at $t'$, i.e.\ marked partitions that do not contain $v$ to begin with.
Additionally, we can also remove $v$ from marked partition at $t'$, by either removing $v$ from its partition set or, if $v$ was singleton, dropping the block from the partition.
In the latter case, we use the marked partitions from $t'$ with the presignature $(X, L, c-1)$, as detaching $v$ creates one additional component that is counted in $c$ at $t$.
Formally, we define
\begin{align*}
    dp[t, X, L, c] &= dp[t', X, L, c] \\
             &\quad \mincup \mathtt{proj}(\set{v}, dp[t', X \cup \set{v},L, c]) \\
             &\quad \mincup \opdetach(\set{v}, dp[t', X \cup \set{v},L, c-1]).
\end{align*}

\paragraph*{Join Node}
Suppose that $t$ is a join node with children $t'$ and $t''$. We want to consider all combinations of solution signatures at $t'$ and $t''$ by combining their compatible partial solutions $H'$ and $H''$. We define
\[
    dp[t, X, L, c] = \bigmincup_{L'\subseteq X, c'\leq c} \mathtt{join}(dp[t', X, L', c'], dp[t'', X, L\symdiff L', c-c']).
\]

\begin{lemma}\label{lm_depot_compatible_solution}
    For every node $t$, presignature $(X, L, c)$, and marked partition $(\P, \mathcal{B}, w) \in dp[t, X, L, c]$, there is a partial solution compatible with $(X, L, c, \P, \mathcal{B})$ at $t$ of cost at most $w$.
\end{lemma}
\begin{proof}
    We prove the statement by bottom-up induction on the tree decomposition. 
    
    A leaf node only has presignature $(\emptyset,\emptyset,0)$ with marked partition $(\emptyset, \emptyset, 0)$, which is compatible with the empty graph. We now suppose that the lemma holds for all child nodes of a node $t$.
    
    Let $t$ introduce vertex $v$ with child $t'$.
    If $v \notin X$ and $v \notin L$, every $(\P,\mathcal{B},w) \in dp[t,X,L,c]$ is also contained in $dp[t',X,L,c]$, where it has a compatible partial solution $H$ with a cost of at most $w$. $H$ remains compatible to $(X,L,c,\P,\mathcal{B})$ at $t$.
    If $v \in X$ and $v \notin L$, let $H$ be the partial solution compatible with $(X\setminus\set{v},L,c,\P\setminus\set{\set{v}}, \mathcal{B}\setminus\set{\set{v}})$ in $t'$ and a cost of at most $w$. Adding $v$ to $H$ as an isolated vertex creates a solution which is compatible to $(X,L,c,\P,\mathcal{B})$ and has the same cost as $H$. 
    Otherwise $dp[t,X,L,c]= \emptyset$.

    If $t$ introduces the edge $\set{u,v}$ with child node $t'$, two cases are possible. If $(\P, \mathcal{B},w)\in dp[t',X,L,c]$ there exists a partial solution $H$ with a cost of at most $w$, which is compatible with the corresponding solution signature at $t'$ and thus also at $t$. 
    Else, $\set{u,v}$ must have been used to join two sets which previously contained $u$ and $v$ using the $\opglue$ operation. Let $w' =w-w(\set{u,v})$ and $(\P',\mathcal{B}',w')\in dp[t',X,L\symdiff\set{u,v},c]$ such that $\opglue_w({\set{u,v}},\set{(\P',\mathcal{B}',w')}) = \set{(\P,\mathcal{B},w)}$. We know that there is a partial solution $H'$ with a cost of at most $w'$ compatible with $(X,L\symdiff\set{u,v},c,\P',\mathcal{B}')$ at $t'$. When adding the edge $\set{u,v}$ to $H$, this creates a partial solution $H$ with a cost of at most $w$. By the definition of $\opglue$, $H$ is compatible with $(X,L,c,\P,\mathcal{B})$ at $t$.

    Let $t$ be a forget node with child $t'$ and $V_t = V_{t'} \setminus \set v$. 
    If $(\P, \mathcal{B}, w)\in dp[t',X,L,x]$, there exists a partial solution with a cost of at most $w$ that is compatible with $(X,L,c,\P,\mathcal{B})$ at $t'$ and thus also at $t$.
    If there exists a marked partition $(\P',\mathcal{B}',w) \in dp[t', X \cup \set{v},L,c]$ such that $\mathtt{proj}(\set{v}, \set{(\P',\mathcal{B}',w)}) = \set{\P,\mathcal{B},w}$, let $H$ be the partial solution compatible with $(X\cup \set{v},L,c,\P',\mathcal{B}')$ at $t'$ and with a cost of at most $w$. By the definition of $\opproj$, we know that the connected component containing $v$ in $H$ also contains an $u \in X_t$, thus $H$ is also compatible with $(X,L,c,\P,\mathcal{B})$ at $t$.
    Else, there exists a marked partition $(\P',\mathcal{B}',w) \in dp[t', X \cup \set{v},L,c-1]$ such that $\opdetach(\set{v}, \set{(\P',\mathcal{B}',w)}) = \set{(\P,\mathcal{B},w)}$. Let $H$ be the partial solution with a cost of at most $w$ compatible with $(X\cup \set{v},L,c-1,\P',\mathcal{B}')$ at $t'$. In contrast to the case above, we now know that the connected component containing $v$ is completely detached from $X_t$ and contains a depot. Thus, $H$ is also compatible with $(X,L,c,\P,\mathcal{B})$ at $t$.

    Lastly, suppose that $t$ is a join node with children $t'$ and $t''$.
    Let $L' \subseteq X$ and $c' \le c$ so that $(\P,\mathcal{B},c) \in \opjoin(dp[t',X,L',c'], dp[t'',X,L\symdiff L',c-c'])$.
    Then there must exist $(\P',\mathcal{B}',w')\in dp[t',X,L',c']$ and $(\P'',\mathcal{B}'',w'')\in dp[t'',X,L\symdiff L',c-c']$ such that $(\P,\mathcal{B},c)\in \opjoin(\set{(\P',\mathcal{B}',w')},\set{(\P'',\mathcal{B}'',w'')})$.
    We take $H'$ and $H''$ to be the partial solutions compatible with $(\P',\mathcal{B}',w')$ and $(\P'',\mathcal{B}'',w'')$ with a cost of at most $w'$ and $w''$ and claim that $H=H' + H''$ is compatible with $(X,L,c,\P,\mathcal{B})$ at $t$.
    First, we have that the weight of $H$ is at most $w = w' + w''$.
    As every edge is only introduced once, we have $H' + H'' \subseteq G_V^\downarrow$.
    Additionally, this implies that vertices have an odd degree if and only if they also have an odd degree in either $H'$ or $H''$, so $L' \symdiff L'' = L$.
    With the same argument we can verify that $H$ has $c'+c''$ connected components which are separate from $X_t$.
    By the definition of $\opjoin$ the other connected components are correctly represented in $\P$ and $\mathcal{B}$.
    Since, further, $C \cap V_t^\downarrow \subseteq V(H)$ and $V(H) \cap X_t = X$ still hold, we conclude that $H$ is compatible with $(X,L,c,\P,\mathcal{B})$ at $t$.
\end{proof}

\begin{lemma}
    For every node $t$, presignature $(X, L, c)$, partition $\P$ and set $\mathcal{B}$ such that there is a partial solution $H$ compatible with $(X, L, c, \P, \mathcal{B})$ of cost $w_H$  at $t$, there is a marked partition $(\P, \mathcal{B}, w)$ in $dp[t, X, L, c]$ with $w \le w_H$.
\end{lemma}
\begin{proof}
    We prove the lemma by induction on the tree decomposition starting at the leaf nodes, similarly to the proof of \Cref{lm_depot_compatible_solution}.
    
    Since the leaf nodes contain no vertex, the only valid presignature for some leaf node $t$ is $(\emptyset, \emptyset, 0)$, the only partial solution compatible with this presignature is the empty graph, which has a cost of 0. Since $dp[t,\emptyset,\emptyset,0]=\set{(\emptyset, \emptyset,0)}$, the lemma holds true for $t$.
    We suppose that the lemma holds for every child node of a node $t$. 
    
    Let node $t$ with child $t'$ introduce vertex $v$. 
    If $v \notin X$, $H$ is also compatible with $(X,L,c,\P, \mathcal{B})$ at $t'$.
    Thus, by the induction hypothesis, for some $w\leq w_H$ we know that $(\P, \mathcal{B}, w)\in dp[t',X,L,c]$.
    Since, in this case, $dp[t',X,L,c]=dp[t,X,L,c]$ we get $(\P, \mathcal{B}, w)\in dp[t,X,L,c]$.
    If $v \in X$, we know that $v$ is isolated in $H$ and $v \notin L$.
    We also know that $(V(H)\setminus\set{v},E(H))$ is compatible with $(X\setminus\{v\},L,c,\P\setminus\set{\set{v}}, \mathcal{B}\setminus\set{\set{v}})$ at $t'$ and thus $(\P\setminus\set{\set{v}},\mathcal{B}\setminus\set{v}, w) \in dp[t',X\setminus\set{v},L,c]$ for some $w\leq w_H$. By definition of $\opins$ we know that $(\P, \mathcal{B}, w)\in dp[t,X,L,c]$.
    
    Suppose that $t$ is an introduce edge node with child $t'$ and $E_t = E_{t'} \cup \set{u,v}$. 
    Again, we need to consider two cases. First, if $\set{u,v}\notin E(H)$, $H$ is compatible with $(X,L,c,\P,\mathcal{B})$ at $t'$ and thus, for some $w\leq w_H$, we know that $(\P, \mathcal{B}, w)\in dp[t',X,L,c]$. By definition, this also implies that, for some $w' \leq w$ we have $(\P, \mathcal{B}, w')\in dp[t,X,L,c]$.
    If $\set{u,v}\in E(H)$, there is a partial solution $H' = (V(H),E(H)\setminus\set{\set{u,v}})$ with cost $w_{H'} = w_H - w(\set{u,v})$ compatible with $(X,L\symdiff \set{u,v},c,\P', \mathcal{B}')$ at $t'$ for some $\P'$ and $\mathcal{B}'$. If $u$ and $v$ are connected in $H'$ we know that $\P'=\P$ and $\mathcal{B}' = \mathcal{B}$, otherwise $\P'$ is such that the blocks $S,T \in \P'$, which contain $u$ and $v$ are joined in $\P$, and all other blocks remain the same. In this case $\mathcal{B} \setminus \set{S \cup T} = \mathcal{B}'\setminus \set{S,T}$ and $\mathcal{B}'$ contains $S$ and $T$ if there exists a depot in the connected component of $S$ or $T$ in $H'$. We know that there exists a $w \leq w_{H'}$ so that $(\P',\mathcal{B}',w)\in dp[t',X,L\symdiff \set{u,v},c]$. Therefore, there is a $w' \leq w_H$ such that $(\P,\mathcal{B},w') \in \mathtt{glue}_{w}(\set{u,v}, dp[t', X, L \symdiff \set{u,v}, c])$ and $(\P,\mathcal{B},w')\in dp[t, X, L, c]$. 

    Let $t$ with child $t'$ forget vertex $v$. 
    If $v \notin V(H)$, then $H$ is also compatible with the solution signature $(X, L, c, \P, \mathcal{B})$ at $t'$. Thus there exists $(\P,\mathcal{B},w) \in dp[t', X, L, c]$ with $w \leq w_H$ and $(\P,\mathcal{B},w) \in dp[t, X, L, c]$. 
    Otherwise, $v \in V(H)$. If $v$ is connected to some other vertex $u \in X$, $H$ is compatible with $(X \cup \set v, L, c, \P', \mathcal{B})$, where $\P' = \set{u,v} \sqcup \P$. Thus, for some $w \leq w' \leq w_H$ we have $(\P',\mathcal{B},w')\in dp[t',X\cup\set{v},L,c]$ and, by the definition of $\opproj$,$(\P,\mathcal{B},w)\in dp[t',X\cup\set{v},L,c]$. For the last case, that is $v \in V(H)$ and $v$ is not reachable from any vertex in $X$, $v$ is contained in a connected component which is separated from $X$. At $t'$, this connected component is still active and $H$ is compatible with $(X\cup\set{v},L,c-1,\P\cup\set{\set{v}},\mathcal{B}\cup\set{\set{v}})$. By the induction hypothesis, there exists $w\leq w' \leq w_H$ so that $(\P\cup\set{\set{v}},\mathcal{B}\cup\set{\set{v}},w)\in dp[t',X\cup\set{v},L,c-1]$ and $(\P,\mathcal{B},w)\in dp[t,X,L,c]$.     

    Lastly, suppose that $t$ is a join node with children $t'$ and $t''$.
    Let $H_{t'} = H \cap G_{t'}^\downarrow$ and $H_{t''} = H \cap G_{t''}^\downarrow$ with costs $w_{t'}$ and $w_{t''}$, respectively.
    Since the edges of $G_t^\downarrow$ are split into $G_{t'}^\downarrow$ and $G_{t''}^\downarrow$ we have $w = w_{t'} + w_{t''}$.
    Let $L_{t'}, L_{t''}\subseteq X$ be the sets of vertices with an odd degree and $c_{t'}$ and $c_{t''}$ be the number of connected components separated from $X$ in $H_{t'}$ and $H_{t''}$.
    We know that $L = L_{t'}\symdiff L_{t''}$ and $c_{t'}+c_{t''}=c$.
    Let $\P_{t'}$ be a partitions of $X$ corresponding to the connected component of $H_{t'}$ and $\mathcal{B}_{t'}\subseteq X$ contain the connected components connected to a depot.
    We know that $H_{t'}$ is compatible with $(X,L_{t'},c,\P_{t'},\mathcal{B}_{t'})$.
    Thus, for some $w'_{t'}\leq w_{t'}$ we have $(\P_{t'},\mathcal{B}_{t'},w'_{t'}) \in dp[t',X,L_{t'},c_{t'}]$.
    Defining $\P_{t''},\mathcal{B}_{t''}$ and $w'_{t''}\leq w_{t''}$ analogously for $ H_{t'}$ yields $(\P_{t''},\mathcal{B}_{t''},w'_{t''}) \in dp[t'',X,L_{t''},c_{t''}]$.
    By definition, for some $w\leq w' \leq w_H$ we have $(\P,\mathcal{B},w')\in\opjoin(\set{(\P_{t'},\mathcal{B}_{t'},w'_{t'})},\set{(\P_{t''},\mathcal{B}_{t''},w'_{t''})})$ and $(\P,\mathcal{B},w)\in dp[t,X,L,c]$.
\end{proof}

The above algorithm yields:

\begin{theorem}
    \label{thm:uncapacitated-tw-fpt}
    $\VRP$ parameterized by $\tw$ is in \FPT.
\end{theorem}
\begin{proof}
    We first obtain a nice tree decomposition $\T$ of $G$ with width $\O(\tw)$ in $\FPT$-time with respect to $\tw$ using \Cref{thm:tw-2-approx,thm:nice-tree-decomposition-transform}.
    Next, we compute the dynamic programming array $dp$ according to \Cref{def:vrp-dp}.
    The array has a total of $2^{\O(\tw)}\cdot n ^{\O(1)}$ entries.
    Every entry contains at most $\tw^{\O(\tw)}\cdot n ^{\O(1)}$ marked partitions, since we only keep marked partitions with optimal weight.
    For the computation of the entry at any presignature $(X,L,c)$ at a node $t$, we consider the marked partitions of the sole child node of $t$, or all combinations of marked partitions at the two child nodes, if $t$ is a join node. Thus, every node requires a running time of at most $\tw^{\O(\tw)}\cdot n ^{\O(1)}$.
    We then check the weight of all marked partitions at the root node $r$ of $\T$ in $dp[r, \emptyset, \emptyset, \emptyset, k']$ for all $k' \le k$ and accept or reject the instance accordingly. In conclusion, the described algorithm has a total running time in $\tw^{\O(\tw)}\cdot n ^{\O(1)}$.
\end{proof}

\section{Capacitated Vehicle Routing}
\label{sec:capacitated}

In this section, we investigate the complexity of $\LoadCVRP$, $\GasCVRP$, and $\LoadGasCVRP$.
We observe that the three problems share most complexity characteristics.
In particular, the requirement to adhere to vehicle capacities greatly increases the complexity.
Whereas $\VRP$ admits an $\FPT$ algorithm when parameterized by the treewidth of $G$, all problems investigated in this section are $\NP$-hard even on trees.
We derive this hardness by encoding $\BinPacking$ instances in vehicle routing problems.
In doing so, we also realize that the presence of zero-weight edges greatly impacts the tractability when parameterizing by the output weight $r$.

Next, we investigate the parameterized complexity with respect to the given vehicle capacities ($\ell$ for $\LoadCVRP$ and $g$ for $\GasCVRP$).
We show $\paraNP$-hardness by reduction from $\TrianglePacking$ for all three variants.
This shows that capacitated vehicle routing problems do not only pose the challenge of ``algebraic packing'' as in $\BinPacking$, but also one of ``structural packing''.
We further provide $\paraNP$-hardness results for $\LoadGasCVRP$, when parameterized by treewidth, number of depots and either the load capacity or gas constraint, both by reduction from $\NTDM$. 

Finally, we find that the situation is not as dim when parameterizing by both the treewidth and the capacity value (for example $\tw + \ell$ for $\LoadCVRP$).
While we cannot yet answer whether or not there can exist an $\FPT$ algorithm for these parameters, we provide an $\XP$ algorithm for $\LoadGasCVRP$ parameterized by $\tw + \ell + g$, thereby ruling out $\paraNP$-hardness.
Additionally, our (to the best of our knowledge) novel $\FPT$-algorithm for $\BinPacking$ parameterized by the capacity of each bin is a solid starting point for further analysis of the parameterized complexity of capacitated vehicle routing.

We start by extending some of our results from \Cref{sec:uncapacitated} to capacitated vehicle routing.

\begin{theorem}
    \label{thm:cvrp-c-fpt}
    For each of $\LoadCVRP$, $\GasCVRP$, and $\LoadGasCVRP$, there is an algorithm that computes an optimal routing in $\cardinality{C}^{\O(\cardinality C)} \cdot n^{\O(1)}$ steps.
\end{theorem}
\begin{proof}
    The algorithm is an extension of the algorithm from \Cref{thm:vrp-c-fpt}.
    Whenever a $\VRP$ routing is generated, it should only be considered as a feasible $\CVRP$ solution if it satisfies $\Gamma_{\ell, \Lambda}$ and/or $\Gamma_g$ (according to which problem is being solved).
    This can be checked in $n^{\O(1)}$ steps per $\VRP$ routing.
\end{proof}

Note that the proof for \Cref{thm:vrp-r-fpt} does not work for $\LoadCVRP$ in general, as there is no sensible way to update the demand $\Lambda$ when contracting zero-weight edges between two clients.
In particular, before merging vertices $u, v$, the two clients could be covered by different vehicles, but after the merge, they need to be covered by the same vehicle.
This problem also occurs in $\LoadGasCVRP$.
However, as the assignment of clients is unimportant for $\GasCVRP$, the proof of \Cref{thm:vrp-r-fpt} works analogously for the following statement.
\begin{corollary}
    \label{thm:gas-cvrp-r-fpt}
    There is an algorithm that, given a $\GasCVRP$ instance and $r \in \N$, decides whether there exists a routing of weight at most $r$ in $f(r) \cdot n^{\O(1)}$ steps, for some computable function $f$.
\end{corollary}
When no zero-weight edges exist, the proof of \Cref{thm:vrp-r-fpt} also suffices, as the step of contracting zero-weight edges can be skipped.
\begin{corollary}
    \label{thm:load-cvrp-no-zero-r-fpt}
    There is an algorithm that, given a $\LoadCVRP$ or $\LoadGasCVRP$ instance with no zero-weight edges and $r \in \N$, decides whether there exists a routing of weight at most $r$ in $f(r) \cdot n^{\O(1)}$ steps, for some computable function $f$.
\end{corollary}

Another insight is that limiting both $k$ and the capacity constraint suffices for tractability, because the number of clients that can possibly be covered is limited.

\begin{theorem}
    \label{thm:cvrp-cap-k-fpt}
    All of the following problems are in $\FPT$:
    \begin{itemize}
        \item $\LoadCVRP$ parameterized by $k + \ell$,
        \item $\GasCVRP$ parameterized by $k + g$, and
        \item $\LoadGasCVRP$ parameterized by $k + \ell$.
    \end{itemize}
\end{theorem}
\begin{proof}
    For $\LoadCVRP$ and $\LoadGasCVRP$, the vehicles can in total cover most $k \cdot \ell$ clients, as $\Lambda(c) \ge 1$ for all $c \in C$.
    Thus, if $k \cdot \ell < \cardinality C$, the instance can be rejected.
    Otherwise, $\cardinality C \le k \cdot \ell$ and the algorithm from \Cref{thm:cvrp-c-fpt} can be used.

    For $\GasCVRP$, we first contract every zero-weight edges.
    Afterwards, the vehicles can in total cover at most $k \cdot g$ clients, because a vehicle has to traverse an edge with positive weight to enter a vertex and we assume the weights to be integral.
    The rest of the argument is analogous to $\LoadCVRP$.
\end{proof}

\subsection{Hardness from Bin Packing}
\label{sec:bin-packing}

Prior work regarding capacitated vehicle routing, such as by \cite{ralphs2003capacitated}, has already observed the connection to $\BinPacking$.
Before we can formally capture the inclusion of $\BinPacking$ within vehicle routing, we first introduce the relevant prior work.

\begin{definition}[$\BinPacking$]
    A $\BinPacking$ instance consists of a finite set $U$ of items, with sizes given by $s\colon U \to \N^+$, a bin capacity $B \in \N^+$, and a number of bins $k \in \N$.
    It asks whether there exists a partition $U_1, U_2, \dots, U_k$ of $U$ such that for all $i \in [k]$ we have $\sum_{u \in U_i} s(u) \le B$.

    Additionally, we define $\UnaryBinPacking$ as a $\BinPacking$ variant, in which all numerical inputs are encoded in unary.
\end{definition}

Going forward, we assume that $\BinPacking$ instances do not contain items $u \in U$ with $s(u) > B$, as these instances can be rejected immediately.

\begin{theorem}[\cite{garey1979computers}]
    Both $\BinPacking$ and $\UnaryBinPacking$ are $\NP$-complete.
\end{theorem}

\begin{theorem}[\cite{jansen2013bin}]
    $\UnaryBinPacking$ parameterized by the number of bins $k$ is $\W[1]$-hard.
\end{theorem}

\begin{theorem}
    \label{thm:cvrp-k-w1-on-trees}
    $\LoadCVRP$, $\GasCVRP$, and $\LoadGasCVRP$ are $\W[1]$-hard when parameterized by $k$, even with unit demands, unit weights, a single depot, and on trees.
\end{theorem}
\begin{proof}
    Similar to before, we give an extensive proof for $\LoadCVRP$ and briefly explain why it carries over to the other variants.
    Given a $\UnaryBinPacking$ instance $(U, s, k, B)$, we construct a tree $G$ with one central depot $d$ and one branch for every item $u \in U$.
    To build such a branch, we add $s(u)$ vertices that form a path $v_u^1 v_u^2 \dots v_u^{s(u)}$ and connect this path to the depot via an edge $\set{d, v_u^1}$.
    The set of clients is the set of all vertices on these paths, that is, $C = \set{v_u^i \with u \in U \land i \in [s(u)]}$.
    We then output the $\LoadCVRP$ instance
    \[
        (G, w_1, D = \set d, C, k, \Lambda_1, \ell = B, r = 2 \cardinality{E(G)}),
    \]
    where $w_1$ and $\Lambda_1$ are the unit weight and unit demand functions $w_1\colon e \mapsto 1$ and $\Lambda_1\colon c \mapsto 1$ respectively.
    \Cref{fig:cvrp-k-reduction-example} shows an example of the reduction.
    Note that even though we are creating $s(u)$ vertices for every $u \in U$, this reduction still runs in polynomial time, as we are reducing from $\UnaryBinPacking$.

    \begin{figure}[ht]
        \centering
        \includegraphics{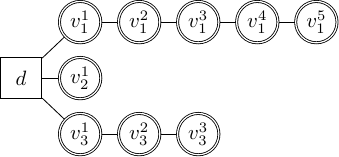}
        % \tikz [layered layout]
        %     \graph [VRP instance, grow'=right] {
        %         "$d$"[depot] -- {
        %             [nodes=client]
        %             "$v_1^1$" -- "$v_1^2$" -- "$v_1^3$" -- "$v_1^4$" -- "$v_1^5$",
        %             "$v_2^1$",
        %             "$v_3^1$" -- "$v_3^2$" -- "$v_3^3$",
        %         }
        %     };
        \caption{Example graph for $U = [3]$ with $s(1) = 5$, $s(2) = 1$, and $s(3) = 3$.}
        \label{fig:cvrp-k-reduction-example}
    \end{figure}

    If there is a valid partition of $U$ into $k$ sets $U_1, \dots, U_k$, which serves as a solution to the $\BinPacking$ instance, for every $i \in [k]$ we let a vehicle cover all clients $v_u^j$ for $u \in U_i$ and $j \in [s(u)]$.
    This can be achieved by having the vehicle traverse the tree from the depot $d$ to $v_u^{s(u)}$ and back for each $u \in U_i$.
    This way, every edge is walked exactly twice and a vehicle $i \in [k]$ covers $\sum_{u\in U_i} s(u)\leq B$ clients.

    Suppose in turn that there is a $\LoadCVRP$ routing $(R = \set{p_1, p_2, \dots, p_{k'}}, A)$ for the given instance.
    Note that as $G$ is a tree, $r = 2 \cardinality{E(G)}$ is just enough so that the vehicles can reach the vertices $\set{v_u^{s(u)} \with u \in U}$.
    Thus, the path of an item $u \in U$ is traversed by exactly one vehicle.
    Given that $\ell = B$, the partition given by $U_i = \set{u \in U \with v_u^1 \in V(p_i)}$ for all $i \in [k']$ abides the constraints of $\BinPacking$.

    The same reduction works without demands for $g = 2B$, as any vehicle can then again visit at most $B$ clients.
    For $\LoadGasCVRP$ the reduction works with $\ell = B$ and $g = 2B$.
\end{proof}

The above translates to the following result regarding treewidth:
\begin{corollary}
    $\LoadCVRP$, $\GasCVRP$, and $\LoadGasCVRP$ are $\paraNP$-hard when parameterized by $\tw + \cardinality D$ and $\W[1]$-hard when parameterized by $\tw + k + \cardinality D$, even with unit demands and unit weights, where $\tw$ denotes the treewidth of $G$.
\end{corollary}

%Finally, we can also modify the reduction in \Cref{thm:cvrp-k-w1-on-trees} for $\LoadCVRP$ to keep $r$ small.
%\begin{corollary}
%    \label{thm:load-cvrp-w1-r}
%    $\LoadCVRP$ and $\LoadGasCVRP$ are $W[1]$-hard when parameterized by $k + r$, even with unit weights, a single depot, and on trees.
%\end{corollary}
%\begin{proof}
%    The construction is analogous to \Cref{thm:cvrp-k-w1-on-trees}, but only a single client $v_u$ with $\Lambda(v_u) = s(u)$ is added for each $u \in U$.
%    Note that $r = 2\cardinality{E(G)} = 2k$.
%\end{proof}

When mapping $\BinPacking$ instances to vehicle routing, the number of bins corresponds to the number of vehicles, and the bin capacity corresponds to the capacity of each vehicle.
As we see later in \Cref{thm:binpacking-B-fpt}, $\BinPacking$ parameterized by $B$ is in $\FPT$.
Thus, an $\FPT$-reduction from $\BinPacking$ mapping a value from $B$ to some vehicle capacity does not yield any hardness result.
$\CVRP$ instances parameterized by the vehicle capacities could potentially be reduced to $\BinPacking$ parameterized by $B$, but such a reduction is unlikely to be applicable to the general case, as vehicle routing problems not only require a numerical partitioning as with $\BinPacking$, but also a structural one found in $\TrianglePacking$.
Still, we return to this intuition in \Cref{sec:cvrp-xp}, where we develop an $\XP$ algorithm for $\LoadGasCVRP$.

\subsection{Hardness from Triangle Packing}
\label{sec:triangle-packing}

In this section, we show that $\TrianglePacking$ is contained within the three $\CVRP$ variants we consider.

\begin{definition}[$\TrianglePacking$]
    A $\TrianglePacking$ instance consists of an undirected graph $G$ such that $\cardinality{V(G)} = 3q$ vertices, for some integer $q \in \N$.
    It asks whether there exists a partition $V_1, V_2, \dots, V_q$ of $V(G)$ so that for each $i \in [q]$, $G[V_i] \simeq K_3$.
\end{definition}
\begin{theorem}[\cite{garey1979computers}]
    $\TrianglePacking$ is $\NP$-complete.
\end{theorem}

Observe that a partition into triangle graphs is equivalent to a fleet of vehicles that each drive a small tour of just three vertices.
We use this intuition in the following reduction.

\begin{theorem}
    \label{thm:cvrp-cap-paraNP}
    All of the following problems are $\paraNP$-hard, even with unit demands and weights:
    \begin{itemize}
        \item $\LoadCVRP$ parameterized by $\ell$,
        \item $\GasCVRP$ parameterized by $g$, and
        \item $\LoadGasCVRP$ parameterized by $\ell + g$.
    \end{itemize}
\end{theorem}
\begin{proof}
    First, we show that $\LoadCVRP$ is $\NP$-hard to decide, even when $\ell = 3$.
    Given a $\TrianglePacking$ instance $G$ with $\cardinality{V(G)} = 3q$, we output the $\LoadCVRP$ instance
    \[
        (G, w_1, D = V(G), C = V(G), k = q, \Lambda_1, \ell = 3, r = 3q),
    \]
    where $w_1$ and $\Lambda_1$ are the unit weight and unit demand functions $w_1\colon e \mapsto 1$ and $\Lambda_1\colon c \mapsto 1$ respectively.

    It is easy to check that a valid triangle packing represents a valid routing, where each triangle is traversed by one vehicle starting at an arbitrary vertex of the triangle.

    Looking at a specific vehicle in any valid routing it can never cover more than $3$ clients, due to the limited load.
    We can also see that it must always cover at least $3$ clients, since otherwise there would be no way to cover $3q$ clients with $q$ vehicles.
    The same argument guarantees that the walks of all vehicles need to be disjoint.
    Therefore the walks of the vehicles form disjoint triangles.
    Since $C = V(G)$, these triangles cover all vertices and form a valid triangle packing of $G$.

    For $\GasCVRP$ and $\LoadGasCVRP$, it is easy to check that the same reduction works with the additional parameter $g = 3$.
    In particular, even without load constraints, $g = 3$ enforces that each vehicle's walk follows a triangle.
\end{proof}

\subsection{Hardness from Numerical 3D Matching}
\label{sec:NTDM}
In this section we use a reduction from $\NTDM$ to show that $\LoadGasCVRP$ is strongly \NP-hard even on trees with constant number of depots, load capacity and unit demands. The reduction can also be adapted to show hardness on trees with constant number of depots, a constant gas constraint and unit edge weights. 

\begin{definition}[\NTDM]
    A \emph{Numerical 3-Dimensional Matching} instance consists of three disjoint sets $X,Y,Z \subseteq \N$, each containing $m$ elements, and a bound $b \in \N$ with $\sum_{a\in X\cup Y \cup Z}a = m\cdot b$. It asks wether $X\cup Y\cup Z$ can be partitioned into $m$ disjoints sets $A_1,\dots, A_m$, such that each $A_i$ contains exactly one element from each of $X,Y,Z$ and $\sum_{a\in A_i}a = b$.    
\end{definition}
\begin{theorem}[\cite{garey1979computers}]
    $\NTDM$ is strongly \NP-complete. 
\end{theorem}
\begin{theorem}\label{thm:loadgascvrp_np_const_demand}
    $\LoadGasCVRP$ is strongly \NP-hard, even on stars with one depot, constant load capacities and unit demands. 
\end{theorem}
\begin{proof}
    Given a $\NTDM$ instance $(X,Y,Z,b)$ we build a star $G$ with the sole depot $d$ as the center. For every $x \in X$ we create a client $v_x$ connected to $d$ by an edge with weight $2^6 \cdot x +2^0\cdot 1 $. Analogously, we create a client $v_y$ and $v_z$ for every $y \in Y$ and $z \in Z$, connected to $d$ by edges of weight $2^6 \cdot y +2^2\cdot 1$ and $2^6 \cdot z +2^4\cdot 1$, respectively. With this, we encode the origin set of the clients in the last six bits of the edge weight. We set $g = 2(2^6 \cdot b +21)$, $\ell = 3$ and define $\Lambda\colon c \mapsto 1$. To conclude the construction of the $\LoadGasCVRP$ instance, we choose $k = m$ and $r =k\cdot g$.   
    An example of this reduction is depicted in \Cref{fig:loadgascvrp-ntdm-example}.

    \begin{figure}[ht]
        \centering
        \includegraphics{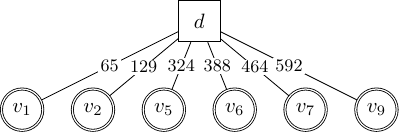}
        % \tikz
        %     \graph [layered layout, VRP instance, sibling distance=12mm, level distance=15mm, edge quotes={fill=white,inner sep=1pt,anchor=center,font=\footnotesize}] {
        %         d/"$d$"[depot] -- {
        %             [nodes=client]
        %             "$v_1$" [>"$65$"],
        %             "$v_2$" [>"$129$"],
        %             "$v_5$" [>"$324$"],
        %             "$v_6$" [>"$388$"],
        %             "$v_7$" [>"$464$"],
        %             "$v_9$" [>"$592$"],
        %         },
        %     };
        \caption{Example for $X = \set{1, 2}$, $Y = \set{5, 6}$, $Z = \set{7, 9}$, and $b = 15$, resulting in $g = 1962$.}
        \label{fig:loadgascvrp-ntdm-example}
    \end{figure}

    When presented with a valid solution $A_1,\dots , A_k$ to the $\NTDM$ instance, for every $i \in [m]$ we let the $i^\text{th}$ vehicle cover the clients $\set{v_x,v_y,v_z \with \set{x,y,z}= A_i}$ in an arbitrary order, driving back to $d$ between visits. Every vehicle visits exactly three clients, and in doing so, accumulates a total edge weight of $2(2^6 \cdot b +21)$. Since $A_1,\dots, A_m$ partitions $X\cup Y\cup Z$, all clients are visited by one vehicle. The total weight of the constructed solution is $2m(2^6 \cdot b +21)$. 

    Let $(R,B)$ be a routing for the constructed $\LoadGasCVRP$ instance. Since $k = m$, $\ell=3$ and $\abs{C}= 3m$, $R$ needs to contain $k$ paths $p_1,\dots, p_k$ and for any $i \in [k]$ we have $ \abs{B^{-1}(i)}=3$.
    Due to the star structure of $G$, every edge has to be traversed twice to visit all clients. Since $\sum_{e\in E(G)}2w(e) = 2m(2^6 \cdot b +21)=m \cdot g$, for all $i \in [k]$ we have $w(p_i)=g$ and every edge is traversed by exactly one vehicle. The three clients visited by a vehicle need to have distinct types, so that $w(p_i)\equiv 2\cdot 21 \mod 2^7$ holds. Thus, for any $i \in [k]$ we can define $A_i=\set{\set{x,y,z}\with v_x, v_y, v_z \in B^{-1}(i)}$. We know $\sum_{a \in A_i} a = b$ and that $A_1,\dots, A_k$ partitions $X\cup Y \cup Z$, proving it a valid solution to the $\NTDM$ instance. 
\end{proof}
To show the hardness of $\LoadGasCVRP$ on stars with one depot and a constant gas constraint, the construction used in the proof of \Cref{thm:loadgascvrp_np_const_demand} can be modified.
We use client demands instead of edge weights to encode the numbers to be matched and their type in a given $\NTDM$ instance, changing the load and gas constraints accordingly. 
\begin{corollary}\label{cor:loadgascvrp_np_const_gas}
    $\LoadGasCVRP$ is strongly \NP-hard, even on stars with one depot and a constant gas constraint.
\end{corollary}

\subsection{Tractability for Constant Treewidth and Vehicle Capacities}
\label{sec:cvrp-xp}

In this section we investigate $\LoadGasCVRP$ parameterized by $\tw + \ell + g$.
Note that the observation from \Cref{thm:evrp-remove-edge-twice} does not hold once vehicles need to abide capacities, as witnessed in \Cref{fig:cvrp-multi-edges-example}.

\begin{figure}[ht]
    \centering
    \includegraphics{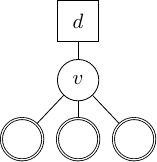}
    % \tikz [layered layout]
    %     \graph [VRP instance] {
    %         "$d$"[depot] -- "$v$" -- {
    %             [nodes=client]
    %             c1/, c2/, c3/
    %         }
    %     };
    \caption{A $\LoadCVRP$ instance with $\ell = 1$ where each routing has to use the edge $\set{d, v}$ more than twice.}
    \label{fig:cvrp-multi-edges-example}
\end{figure}

Thus, the approach of \Cref{def:vrp-dp} cannot directly be applied to $\LoadGasCVRP$.
Instead, we track the number of walks with a given set of characteristics separately, leading to an $\XP$ running time.
Note however, that a similar result to \Cref{thm:evrp-remove-edge-twice} holds for each vehicle individually:

\begin{lemma}
    \label{thm:loadgascvrp-remove-edge-twice}
    Let $(R, A)$ be a $\LoadGasCVRP$ routing, let $p \in R$ and let $H \subseteq G$ be the undirected multigraph representing $p$.
    If there is an edge $e \in E(H)$ with multiplicity greater than two, then the walk $p'$ obtained from $H - 2e$ can be used instead of $p$ in $R$ to form a $\LoadGasCVRP$ routing $R'$ with $w(R') \le w(R)$.
\end{lemma}
\begin{proof}
    Analogously to \Cref{thm:evrp-remove-edge-twice}, with the additional remark that $p'$ covers the same clients as $p$ and that $w(p') \le w(p)$ ensures that the gas capacity is not exceeded.
\end{proof}

Before we formally define the state of the dynamic programming algorithm on the tree decomposition of a given instance for this section, we divert back to $\BinPacking$.
We define and solve an extended version of $\BinPacking$, which abstracts from the combination of different sets of walks in order to form new walks fulfilling some desired characteristics. In the algorithm, this operation will be needed for the computation at the Join Nodes of the tree decomposition.

In particular, we encode each walk as a $\BinPacking$ item with a \emph{multidimensional size} given by the number of clients and distance traveled.
Accordingly, the bins have multidimensional capacities matching the characteristics of the desired walks.
As the walks of a partial solution might have different characteristics, there are several different \emph{bin kinds}.
In addition, each walk has a \emph{fingerprint} given by the endpoints of the walk and whether or not it contains a depot.
Each bin kind expects a combination of fingerprints fulfilling a certain predicate, according to the endpoints of the desired walk.

\begin{upright-definition}[\HetMdFpBinPacking]
    A \emph{Heterogeneous Multidimensional Fingerprint Bin Packing} instance consists of
    \begin{itemize}
        \item a finite set $U$ of items,
        \item a finite set $\mathcal F$ of fingerprints,
        \item a dimension $d \in \N^+$,
        \item a number of bin kinds $m \in \N^+$,
        \item bin capacities $B_1, B_2, \dots, B_m \in \N^d$ (with $\norm{B_1}_1 \ge \norm{B_2}_1 \ge \dots \ge \norm{B_m}_1$),
        \item available bin counts $k_1, k_2, \dots, k_m$,
        \item predicates deciding the validity of a multiset of fingerprints $\mathcal V_1, \mathcal V_2, \dots, \mathcal V_m\colon \N^{\mathcal F} \to \set{0, 1}$,
        \item item sizes $s\colon U \to \N^d \setminus \set 0$ with $\norm{s(u)}_1 \le \norm{B_1}_1$ for all $u \in U$, and
        \item item fingerprints $F\colon U \to \mathcal F$.
    \end{itemize}
    It asks whether there exists a partition $\set{U_{i, j}}_{i \in [m], j \in [k_i]}$ of $U$ such that for all $i, j$
    \begin{itemize}
        \item the bin capacity is not exceeded: $\sum_{u \in U_{i, j}} s(u) \le B_i$, and
        \item the multiset of fingerprints is valid: $\mathcal V_i(\multiset{F(u) \with u \in U_{i, j}}) = 1$. \lipicsEnd
    \end{itemize}
\end{upright-definition}

\begin{theorem}
    \label{thm:hetmdfpbinpacking-fpt-cap}
    There is an algorithm that, given a $\HetMdFpBinPacking$ instance
    \[
        \mathcal B = (U, \mathcal F, d, m, (B_1, B_2, \dots, B_m), (k_1, k_2, \dots, k_m), (\mathcal V_1, \mathcal V_2, \dots, \mathcal V_m), s, F),
    \]
    computes a feasible partition or concludes that there is none in $f(\cardinality{\mathcal F} + d + m + \norm{B_1}_1) \cdot (T(\cardinality U) + \cardinality{\mathcal B}^{\O(1)})$ steps, where $T\colon \N \to \N$ is an upper bound on the time complexity of the predicates $\set{\mathcal V_i}_{i \in [m]}$, for some computable function $f$.
\end{theorem}
\begin{proof}
    First, observe that only the size and the fingerprint of the items $u \in U$ matter for the partitioning, not the identity of each item.
    We call the combination of these two values the \emph{key} of an item.
    Let $\mathcal K = \set{(s(u), F(u)) \with u \in U}$ be the set of all keys present in the input.
    For $u \in U$ and its key $K = (s(u), F(u))$ we denote $s(K) = s(u)$.
    Additionally, for $K \in \mathcal K$, let $q_K$ denote the number of items $u \in U$ with key $K$.

    In the output, each bin contains a certain number of items matching a given key.
    More formally, for each bin kind $i \in [m]$, we call a function $c\colon \mathcal K \to \N$ a \emph{feasible bin configuration} for kind $i$, if a bin with such items passes the criteria of $\HetMdFpBinPacking$.
    We denote the multiset of fingerprints of a configuration $c$ as $F(c)$.
    $F(c)$ contains each fingerprint $f \in \mathcal F$ with multiplicity $\sum_{x} c(x, f)$.
    Then, more formally, the set of all feasible bin configurations for kind $i$ is
    \[
        C_i = \set*{c\colon \mathcal K \to \N \with \sum_{K \in \mathcal K} c(K) \cdot s(K) \le B_i \land \mathcal V_i(F(c)) = 1}.
    \]

    We define an ILP $\mathcal I$ that computes, for each bin kind $i \in [m]$ and configuration $c \in C_i$, the number $b_{i, c}$ of bins that should be configured according to $c$.
    \begin{align*}
        b_{i, c} &\in \N && \forall i \in [m], c \in C_i \\
        \sum_{c \in C_i} b_{i, c} &\le k_i && \forall i \in [m] \\
        \sum_{i \in [m], c \in C_i} b_{i, c} \cdot c(K) &= q_K && \forall K \in \mathcal K \\
    \end{align*}
    A satisfying assignment of the $b_{i, c}$ variables is equivalent to a partition of $U$, where for each $i \in [m]$ and $c \in C_i$ there are $b_{i, c}$ bins of the $i$-th kind configured according to $c$.

    To show that the algorithm terminates in the desired number of steps, we will bound the sizes of the sets $\mathcal K$ and $C_i$ for all $i \in [m]$.
    For the former, observe that as each $u \in U$ has $\norm{s(u)}_1 \le \norm{B_1}_1$,
    \[
        \cardinality{\mathcal K} \le (\norm{B_1}_1 + 1)^d \cdot \cardinality{\mathcal F}.
    \]
    Next, let $i \in [m]$ and $c \in C_i$.
    As all items contribute to a bin's size in at least one dimension, we have $c(K) \le \norm{B_1}_1^d$ for each $K \in \mathcal K$.
    Therefore, the size of the set $C_i'$ of functions $c\colon \mathcal K \to \N$ satisfying the bin capacity condition is bounded by 
    \[
        \cardinality{C_i'} \le \left( \norm{B_1}_1^d + 1 \right)^{\cardinality{\mathcal K}}.
    \]
    Filtering the set $C_i'$ for only those functions that satisfy $\mathcal V_i$ gives the final set $C_i$.
    Thus, the number of invocations of $\mathcal V_i$ is bounded as stated in the theorem.

    Finally, solving $\mathcal I$ can be done in FPT time with respect to the number of variables.
    As there is one variable for each combination of bin kind $i$ and feasible bin configuration $c \in C_i$, this number is bounded as desired.~\cite{lenstra1983integer, kannan1987minkowski, frank1987application, cygan2015parameterized}
\end{proof}

Note that removing all the additions of $\HetMdFpBinPacking$, the above algorithm can also be used for regular $\BinPacking$.

\begin{corollary}
    \label{thm:binpacking-B-fpt}
    $\BinPacking$ parameterized by the bin capacity $B$ is in $\FPT$.
\end{corollary}

With all preliminaries out of the way, we present the dynamic programming algorithm for $\LoadGasCVRP$.
Given a $\LoadGasCVRP$ instance $(G, w, D, C, k, \ell, \Lambda, g, r)$ and a nice tree decomposition $\T$ of $G$, we track partial solutions of the following form:

\begin{upright-definition}
    \label{def:loadgascvrp-dp}
    Let $t \in \T$.
    We call a walk $p$ in $G_t^\downarrow$ \emph{detached} if $V(p) \cap X_t = \emptyset$ and $V(p) \cap D \ne \emptyset$.
    We call $p$ \emph{attached} when both endpoints of $p$ are in $X_t$ (but we don't require anything with regards to $D$).
    For $u, v \in V(G), \lambda, w \in \N$, an attached $u$-$v$-walk with $\Lambda(p) = \lambda$ and $w(p) = w$ is called a $(u, v, \lambda, w)$-walk.

    Let $S_t = \set{s\colon X_t^2 \times [\ell]_0 \times [g]_0 \to \N}$ be the set of all multisets of combinations of walk endpoints, clients covered, and weight.
    We refer to the elements $s \in S_t$ as \emph{counters} and abbreviate $\Lambda(s) = \sum_{u, v, \lambda, w} \lambda \cdot s(u, v, \lambda, w)$ and $w(s) = \sum_{u, v, \lambda, w} w \cdot s(u, v, \lambda, w)$.

    Let $X \subseteq X_t \cap C$, $c \le k$, and $s, s^* \in S_t$.
    We call $(X, c, s, s^*)$ a \emph{solution signature} at $t$.
    A set of walks $R = \set{p_1, p_2, \dots, p_{k'}}$ in $G_t^\downarrow$ and assignment of clients $A\colon C \cap (X \cup (V_t^\downarrow \setminus X_t)) \to [k']$ is a \emph{partial solution} compatible with $(X, c, s, s^*)$ at $t$, if
    \begin{itemize}
        \item for all clients $c \in C \cap (X \cup (V_t^\downarrow \setminus X_t))$, $c \in V(p_{A(c)})$,
        \item all walks $p \in R$ are either detached or attached,
        \item there are exactly $c$ detached walks in $R$,
        \item
            for each $u, v \in X_t, \lambda \in [\ell]_0, w \in [g]_0$,
            \begin{itemize}
                \item there are exactly $s(u, v, \lambda, w)$ many $(u, v, \lambda, w)$-walks $p \in R$ with $V(p) \cap D = \emptyset$, and
                \item there are exactly $s^*(u, v, \lambda, w)$ many $(u, v, \lambda, w)$-walks $p \in R$ with $V(p) \cap D \ne \emptyset$. \lipicsEnd
            \end{itemize}
    \end{itemize}
\end{upright-definition}

We use dynamic programming to compute the minimum weight $dp[t, X, c, s, s^*] \in \N \cup \set{\infty}$ of all partial solutions $R$ compatible with $(X, c, s, s^*)$ at all $t \in \T$ and for all solution signatures $(X, c, s, s^*)$ at $t$.
The observant reader might have noticed that since the counters $s \in S_t$ map to the natural numbers, the set $S_t$ is infinite when $X_t \ne \emptyset$.
However, we know from \Cref{thm:loadgascvrp-remove-edge-twice} that it suffices to look for routings where each of the $k$ walks contain each edge of $G$ at most twice.
As there are only finitely many such routings, only finitely many partial solutions leading to these routings can be witnessed at the bags of $\T$.
Thus, there are only finitely many counters which are relevant for finding an optimal routing.

First, if the number of covered clients indicated by a counter $s$ is greater than the number of available clients, we know that no compatible partial solution exists.
More formally, if
\[
    \Lambda(s + s^*) > \cardinality{C \cap (X \cup (V_t^\downarrow \setminus X_t))},
\]
then $dp[t, X, \cdot, s, s^*] = \infty$.
Note that the above values might still not be equal, as some of the clients in $V_t^\downarrow$ might be covered by detached walks.
This means that we can limit the computation of $dp$ to only those counters $s \in S_t$ with $s(\cdot, \cdot, \lambda, \cdot) \le \cardinality C$ for all $\lambda \ge 1$.

This still leaves the value of $s$ unbounded for $\lambda = 0$.
However, the $dp$ values for counters with $s(u, v, 0, w) > 1$ can be reduced to the case where $s(u, v, 0, w) = 1$.
The fact that a partial solution $(R, A)$ containing a $(u, v, 0, w)$-walk $p$ exists suffices to find the minimum weight partial solution for $s(u,v,0,w) > 1$.
Such a solution can be obtained by repeatedly inserting copies of $p$ into $R$, which increases the weight of $R$ by $(s(u, v, 0, w) - 1) \cdot w$.
Thus, we only need to compute $dp[t, X, c, s, s^*]$ when $s(\cdot, \cdot, 0, \cdot), s^*(\cdot, \cdot, 0, \cdot) \le 1$.

Before we can start to describe the computation of $dp$, we need to introduce the concept of \emph{counter reducibility}.

\begin{definition}
    Let $t \in T$ and $s_1, s_1^*, s_2, s_2^* \in S_t$.
    The pair of counters $(s_1, s_1^*)$ is reducible to $(s_2, s_2^*)$, denoted as $(s_1, s_1^*) \leadsto (s_2, s_2^*)$, if for every $X \subseteq X_t$, $c \le k$, and partial solution $(R_1, A_1)$ compatible with $(X, c, s_1, s_1^*)$ at $t$, there is a partial solution $(R_2, A_2)$ compatible with $(X, c, s_2, s_2^*)$ at $t$ which can obtained by merging walks of $R_1$ and updating $A_1$ accordingly.
\end{definition}
\begin{lemma}
    \label{thm:counter-reducibility-fpt}
    Counter reducibility can be decided in $f(\cardinality{X_t} + \ell + g + z) \cdot n^{\O(1)}$ steps for some function $f$, where $z = \sum_{u, v \in X_t} (s_2(u, v, 0, 0) + s_2^*(u, v, 0, 0))$.
\end{lemma}
\begin{proof}
    First, we check whether $\Lambda(s_1 + s_1^*) = \Lambda(s_2 + s_2^*)$ and $w(s_1 + s_1^*) = w(s_2 + s_2^*)$.
    If this is not the case, we immediately reject the instance.
    Otherwise, we use the algorithm for $\HetMdFpBinPacking$ from \Cref{thm:hetmdfpbinpacking-fpt-cap} with $d = 3$ and $\mathcal F = X_t^2 \times \set{0, 1}$.
    Each item will be one walk from $(s_1, s_1^*)$.
    The three dimensions indicate the number of covered clients, the weight of the walk, and one extra dimension for otherwise zero-sized items.
    The fingerprints indicate the endpoints of the walks and one bit which encodes whether the walk contains a depot.
    For each $u, v \in X_t, \lambda \in [\ell]_0, w \in [g]_0$ we introduce $s_1(u, v, \lambda, w)$ items $i$ with $s(i) = (\lambda, w, \indic(\lambda = w = 0))$ and $F(i) = (u, v, 0)$.
    Similarly, we introduce $s_1^*(u, v, \lambda, w)$ items $i$, but set $F(i) = (u, v, 1)$.
    The walks in $(s_2, s_2^*)$ are modeled by the bin kinds.
    We first cover the walks in $s_2$.
    For each $u, v \in X_t, \lambda \in [\ell]_0, w \in [g]_0$ we add one bin kind $j$ with $B_j = (\lambda, w, z)$.
    A multiset of fingerprints $\multiset{F_1, F_2, \dots, F_x}$ is valid according to $\mathcal V_j$ if the directed graph formed by the walk endpoints contains an Eulerian trail from $u$ to $v$ and none of the fingerprints has the depot-bit set.
    The number of available bin kinds is $k_j = s(u, v, \lambda, w)$.
    For the walks in $s_2^*$, we add equivalent bin kinds, but require in $\mathcal V_j$ that at least one depot-bit is set.
    Note that the capacity of the third dimension in every bin kind allows that the zero-sized walks to be distributed among the bins arbitrarily.
    To satisfy $\norm{B_1}_1 \ge \norm{B_2}_1 \ge \dots \ge \norm{B_m}_1$, we reorder the numbering of the bin kinds accordingly.

    By using the validity predicates $\mathcal V_j$, we ensure that any $\HetMdFpBinPacking$ partition is equivalent to a proper distribution of the walks.
    Additionally, as we have checked before that the given number of clients and weight of walks from $s_1 + s_1^*$ matches $s_2 + s_2^*$, the partition is tight in the first two dimensions and match the requirements of $s_2 + s_2^*$.
    We have $\cardinality{\mathcal F} = 2 \cdot \cardinality{X_t}^2$, the number of bins is bounded by $2 \cdot \cardinality{X_t}^2 \cdot (\ell + 1) \cdot (g + 1)$, and $\norm{B_1}_1 \le \ell + g + z$.
    As all $\mathcal V_j$ predicates can be checked in polynomial time, the total running time matches the desired bound.
\end{proof}

The computation of $dp$ then depends on the node type of $t$:

\paragraph*{Leaf Node}
If $t$ is a leaf node, then $X_t = \emptyset$, so $X$, $s$ and $s^*$ must all be empty too.
As $V_t^\downarrow = \emptyset$, there cannot be any detached walks, so
\[
    dp[t, X, c, s, s^*] = \begin{cases}
        0, &\text{if } c = 0,\\
        \infty, &\text{otherwise}.
    \end{cases}
\]
\paragraph*{Introduce Vertex Node}
Suppose that $t$ is an introduce vertex node with child node $t'$ such that $X_t = X_{t'} \cup \set v$.
As $v$ is an isolated vertex in $G_t^\downarrow$, any attached walk $p$ that contains $v$ must be the empty walk.
Thus, for all $u \ne v$ and $w > 0$ we must have $s(u, v, \cdot, \cdot) = s(v, u, \cdot, \cdot) = s^*(u, v, \cdot, \cdot) = s^*(v, u, \cdot, \cdot) = s(v, v, \cdot, w) = s^*(v, v, \cdot, w) = 0$.
Depending on whether or not $v \in D$, the empty walks at $v$ are counted in $s$ or $s^*$ respectively.
Suppose that $v \notin D$, then $s^*(v, v, \cdot, \cdot) = 0$.
One of the empty walks at $v$ should cover a client if and only if $v$ is a client and contained in $X$.
More formally, for all $\lambda \ge 1$ we have $s(v, v, \lambda, 0) = \indic(\lambda = \Lambda(v) \land v \in X)$.
Lastly, the value $s(v, v, 0, 0)$ can take any value.
If all of the above constraints are satisfied, we set $dp[t, X, c, s, s^*] = dp[t', X \setminus \set v, c, s_{\downarrow X_{t'}}, s_{\downarrow X_{t'}}^*]$.
Otherwise there is no compatible partial solution, and the $dp$ entry is set to $\infty$.
The case $v \in D$ follows analogously with the requirements for $s$ and $s^*$ swapped.

\paragraph*{Introduce Edge Node}
Suppose that $t$ is an introduce edge node labeled with the edge $e = \set{u, v}$ with a child $t'$.
As seen in \Cref{thm:loadgascvrp-remove-edge-twice}, any walk in a minimum weight partial solution compatible with $(X, c, s, s^*)$ at $t$ traverses $e$ at most twice.
Thus, when removing $e$, a walk $p$ gets split into at most three subwalks.
Let $s', s'^* \in S_{t'}$ and $c_e \in \N$ such that $\max\set{c_e, \cardinality{s'} + \cardinality{s'^*}} \le 3 \cdot (\cardinality{s} + \cardinality{s^*})$.
Then a partial solution compatible with $(X, c, s, s^*)$ at $t$ can be constructed from a partial solution compatible with $(X, c, s', s'^*)$ at $t'$ by inserting $c_e$ copies of $e$ if $(s' + c_e \cdot \multiset{(u, v, 0, w(u, v))}, s'^*) \leadsto (s, s^*)$.
Thus, we set
\[
    dp[t, X, c, s, s^*] = \min_{s', s'^*, c_e} dp[t', X, c, s', s'^*] + c_e \cdot w(u, v),
\]
where $s', s'^*, c_e$ range over the values described above.

\paragraph*{Forget Node}
Suppose that $t$ is a forget node with child node $t'$ such that $X_t = X_{t'} \setminus \set v$.
A partial solution compatible with $(X, c, s, s^*)$ can interact with $v$ in several ways.

First, there can be detached walks containing $v$, which are included in the count $c$.
For a given count $c' \le c$ of detached walks that are also detached from $t'$, let $s_v \in S_{t'}$ be a counter which has $s_v(a, b, \cdot, \cdot) = 0$ whenever $a \ne v$ or $b \ne v$ and $\sum s_v(v, v, \cdot, \cdot) = c - c'$.
There are at most $(c - c' + 1)^{(\ell + 1) \cdot (g + 1)}$ such counters $s_v$.
As the detached walks need to contain a depot, we consider the counter $s_v$ as an addition to $s^*$.

Second, the walks counted in $s$ and $s^*$ can contain $v$.
However, any such walk cannot start or end at $v$, as it would otherwise not be attached to $X_t$.
Thus, the counters for $t'$ remain unchanged by this interaction.

In both cases, if $v \in C$, it must be covered regardless of $X$.
Therefore, we set
\[
    dp[t, X, c, s, s^*] = \min_{c', s_v} dp[t', X \cup (C \cap \set v), c', s, s^* + s_v]
\]
where $c$ and $s_v$ range over the counts and counters described above.

\paragraph*{Join Node}
Suppose that $t$ is a join node with children $t_1$ and $t_2$.
A partial solution compatible with $(X, c, s, s^*)$ is composed of partial solutions from both $t_1$ and $t_2$ to cover the clients in $V_{t_1}^\downarrow$ and $V_{t_2}^\downarrow$ respectively.
A walk from a partial solution at $t$ might switch between $V_{t_1}^\downarrow$ and $V_{t_2}^\downarrow$ a number of times, so it is split into multiple walks at the children.
Suppose that a walk $p$ in $G_t^\downarrow$ is split into the walks $p_1, p_2, \dots, p_q$.
First, the clients covered by $p$ are distributed among the $p_i$.
There are at most $\ell$ such client-connecting walks.
All walks $p_i$ that do not cover a client are only needed to connect the start and end of two client-connecting walks.
As there are $\cardinality{X_t}$ vertices that could be reached, at most $\cardinality{X_t}$ walks are needed in between two client-connecting walks.
In total, there are thus at most $\ell$ client-connecting walks and $(\ell + 1) \cdot \cardinality{X_t}$ walks connecting the endpoints of those.
Therefore, we limit the counters $s_1, s_1^* \in S_{t_1}, s_2, s_2^* \in S_{t_2}$ to those which are bounded above by $(\ell + 1) \cdot \cardinality{X_t} \cdot (\cardinality{s} + \cardinality{s^*})$.
Similarly to before, only the counters which satisfy $(s_1 + s_2, s_1^* + s_2^*) \leadsto (s, s^*)$ are considered.
Note that the clients in $X$ are split between the two partial solutions so that the partial solution in $t_1$ covers some subset $X_1 \subseteq X$ and leaves $X \setminus X_1$ to $t_2$'s solution.
Similarly, the number of detached walks is split into $c_1 \le c$ and $c - c_1$.
Bringing it all together, we set
\[
    dp[t, X, c, s, s^*] = \min_{\substack{X_1 \subseteq X, c_1 \le c,\\ s_1, s_1^*, s_2, s_2^*}} dp[t_1, X_1, c_1, s_1, s_1^*] + dp[t_2, X \setminus X_1, c - c_1, s_2, s_2^*],
\]
where $s_1, s_1^*, s_2, s_2^*$ range over the values described above.

\begin{theorem}
    \label{thm:loadgascvrp-xp}
    $\LoadGasCVRP$ parameterized by $\tw + \ell + g$ is in $\XP$.
\end{theorem}
\begin{proof}
    We first obtain a nice tree decomposition $\T$ of $G$ with width $\O(\tw)$ in $\FPT$-time with respect to $\tw$ using \Cref{thm:tw-2-approx,thm:nice-tree-decomposition-transform}.
    Next, we compute the dynamic programming array $dp$ according to \Cref{def:loadgascvrp-dp}.
    Using the restrictions described, we compute $n^{f(\tw + \ell + g)}$ entries of $dp$ for some computable function $f$. We will argue the running time of the computations in every node type for each signature $(X, c, s, s^*)$.
    \begin{itemize}
        \item The running time is constant for every leaf node. 
        \item An introduce vertex node requires computations with a complexity of $\O(\tw + \ell + g)$.
        \item Every introduce edge node has $n^{f(\tw + \ell + g)}$ many combinations of $c_e$ and $s_v$ to check, for some computable function $f$.
            Since we assume that $s(\cdot, \cdot, 0, 0), s^*(\cdot, \cdot, 0, 0) \le 1$, each of those is checked for reducibility in $\FPT$-time with respect to $\tw + \ell + g$ using \Cref{thm:counter-reducibility-fpt}.
        \item Forget nodes require a running time in $n^{\O(\ell + g)}$.
        \item
            For any join node, similarly to the introduce edge node, there are $n^{f(\tw + \ell + g)}$ combinations of the variables $X_1, c_1, s_1, s_1^*, s_2, s_2^*$ to check, for some computable function $f$.
            Each check takes $\FPT$-time with respect to $\tw + \ell + g$ using \Cref{thm:counter-reducibility-fpt}.
    \end{itemize}
    We then check the weight at the root node $r$ of $\T$ in the entries $dp[r, \emptyset, k', \emptyset, \emptyset]$ for all $k' \le k$ and accept or reject the instance accordingly.
\end{proof}

Furthermore, it is easy to see that the algorithm from \Cref{thm:loadgascvrp-xp} can be modified to yield similar results for $\LoadCVRP$ and $\GasCVRP$.
In particular, the counters $S_t$ for each node $t \in \T$ should only track the endpoints of each walk and one of the two constraints.

\begin{corollary}
    \label{thm:loadcvrp-xp}
    \label{thm:gascvrp-xp}
    $\LoadCVRP$ parameterized by $\tw + \ell$ and $\GasCVRP$ parameterized by $\tw + g$ are in $\XP$.
\end{corollary}

\section{Conclusion}

We have presented detailed tractability results for vehicle routing problems regarding the parameters given in the input and the treewidth of the underlying network.
A summary of our results can be found in \Cref{tab:capacitated_overview_parameters}. Note that hardness results can be inherited, as some problems represent strict generalizations of others.

We have shown complete results for all mentioned parameters in the uncapacitated case.
For $\VRP$ treewidth as a parameter suffices for the existence of $\FPT$ algorithms. On the other hand, once any form of capacity constraints is introduced, vehicle routing becomes hard even on trees.
Combining treewidth with other parameters is only sensible for parameters which are not sufficient for $\FPT$ algorithms on their own, namely the capacity, $k$ and $\abs{D}$.
Out of these parameters, only the inclusion of all present capacity parameters yields tractability.

While we also proved complexity bounds of the capacitated vehicle routing variants, some questions remain unanswered.
We have provided an $\XP$ algorithm for treewidth and capacity as a combined parameter. However, the existence of an $\FPT$ algorithm for the problem could neither be ruled out nor proven in this work. Additionally, while the $\W[1]$-hardness for the parameter $\tw + k + \cardinality D$ is likely to rule out an $\FPT$-algorithm, future work could explore $\XP$ algorithms for this parameter.
\emergencystretch=1em

\bibliographystyle{plainurl}
\bibliography{bibliography}

\end{document}